\documentclass[a4paper, final]{article}
\usepackage{a4wide}
\usepackage[utf8]{inputenc}
\usepackage{amsfonts}
\usepackage{amsmath, amsthm}
\usepackage{amssymb} 
\usepackage{enumitem}
\usepackage{microtype}
\usepackage{hyperref}

\theoremstyle{definition}
\newtheorem{theorem}{Theorem}
\newtheorem{lemma}[theorem]{Lemma}
\newtheorem{proposition}[theorem]{Proposition}
\newtheorem{corollary}[theorem]{Corollary}
\newtheorem{example}[theorem]{Example}
\newtheorem{definition}[theorem]{Definition}

\newcommand{\LA}{\text{\upshape{LA}}}

\newcommand{\A}{\mathcal{A}}
\newcommand{\B}{\mathcal{B}}
\newcommand{\C}{\mathcal{C}}
\newcommand{\I}{\mathcal{I}}

\renewcommand{\O}{\mathcal{O}}
\renewcommand{\P}{\mathcal{P}}

\newcommand{\R}{\mathcal{R}}
\renewcommand{\S}{\mathcal{S}}

\renewcommand{\t}{\mathfrak{t}}

\newcommand{\Real}{\mathbb{R}}

\newcommand{\Nat}{\mathbb{N}}

\newcommand{\equals}{=}
\newcommand{\<}{\langle}
\renewcommand{\>}{\rangle}
\newcommand{\ldbrace}{\mbox{$(\!\cdot$}}
\newcommand{\rdbrace}{\mbox{$\cdot\!)$}}

\newcommand{\vars}{\text{\upshape{vars}}}
\newcommand{\consts}{\text{\upshape{consts}}}
\newcommand{\bconsts}{\text{\upshape{bconsts}}}
\newcommand{\fconsts}{\text{\upshape{fconsts}}}
\newcommand{\aconsts}{\text{\upshape{$\alpha$consts}}}

\newcommand{\len}{\text{len}}

\newcommand{\At}{\text{At}}
\newcommand{\hAt}{\widehat{\At}}

\newcommand{\con}{\leftrightharpoons}
\newcommand{\APC}{\text{\upshape{ap}}}

\newcommand{\Ax}{\text{\upshape{Ax}}}

\newcommand{\subst}[2]{\bigl[#1\!\bigm/\!#2\bigr]}

\begin{document}

\title{Bernays-Sch\"onfinkel-Ramsey with Simple Bounds is NEXPTIME-complete}

\author{
	\begin{tabular}{l}
		Marco Voigt\\
		\small\textit{Max Planck Institute for Informatics, Saarland Informatics Campus, Saarbr\"ucken, Germany,}\\
		\small\textit{Saarbr\"ucken Graduate School of Computer Science}
	\end{tabular}
	\and
	\begin{tabular}{l}
		Christoph Weidenbach \\
		\small\textit{Max Planck Institute for Informatics, Saarland Informatics Campus, Saarbr\"ucken, Germany}
	\end{tabular}
}	
\date{June 2015}
\maketitle

\begin{abstract}
	First-order predicate logic extended with linear arithmetic is undecidable, in general. We show that the 
	Bernays-Sch\"onfinkel-Ramsey (BSR) fragment extended with
	linear arithmetic restricted  to simple bounds (SB)  
	is decidable through finite ground instantiation.
	The identified ground instances can be employed to restrict the search space of 
	existing automated reasoning procedures for BSR(SB).
	Satisfiability of BSR(SB) compared to BSR remains NEXPTIME-complete. 
	The decidability result is almost tight because BSR is undecidable if extended with linear difference inequations, 
	simple additive inequations, quotient inequations and multiplicative inequations.
\end{abstract}


\section{Introduction} \label{section:Introduction}

The Bernays-Sch\"onfinkel-Ramsey (BSR) fragment comprises exactly the first-order logic prenex sentences 
with the $\exists^*\forall^*$ quantifier prefix resulting in a CNF where all occurring function symbols are constants. 
Formulas may contain equality.
Satisfiability of the BSR fragment is decidable and NEXPTIME-complete~\cite{Lewis1980}.
Its extension with linear arithmetic is undecidable~\cite{Halpern1991,CoxMT92,Fietzke2012}. 
Our first contribution is refinements to these results where the arithmetic constraints
are further restricted to linear difference inequations $x-y \triangleleft c$, 
simple additive inequations $x+y \triangleleft c$, 
quotient inequations $x \triangleleft c\cdot y$ and 
multiplicative inequations $x \cdot y \triangleleft c$ where $c\in\Real$, and $\triangleleft\in\{<, \leq, =,\neq,\geq, >\}$. 
Under all these restrictions, the combination
remains undecidable, respectively (see Section~\ref{section:UndecidableFragments}). 

On the positive side, we prove decidability of the restriction to arithmetic constraints consisting of 
simple bounds of the form $x\triangleleft c$, where $\triangleleft$ and $c$ are as above. Underlying
the result is the observation that similar to the finite model property of BSR, only finitely many instances of universally quantified clauses with arithmetic simple bounds constraints need to be considered. Our construction is motivated by 
results from quantifier elimination~\cite{Loos1993} and hierarchic superposition~\cite{Bachmair1994b,Kruglov2012,Fietzke2012}. 
For example, consider the following two clauses (we denote clauses in implication form with the arithmetic part syntactically separated from the free part; $\wedge$ and $\vee$ bind stronger than $\rightarrow$)
\begin{center}
$\begin{array}{r@{\;}r@{\;\;\rightarrow\;\;}l}
x_2\neq 5 			&\wedge\; R(x_1)  	& Q(u_1,x_2)\\
y_1<7 \wedge y_2\leq 2	&	  		& Q(c,y_2) \vee R(y_1)
\end{array}$
\end{center}
where the variable $u_1$ shall range over a freely selectable domain, $x_i$, $y_i$ are variables over the reals, and the constant
$c$ addresses an element of the same domain that $u_1$ ranges over. All occurring variables shall be implicitly universally quantified.
Our main result reveals that this clause set is satisfiable if and only if a finite set of ground instances is satisfiable in which
(i) $u_1$ has been instantiated with the constant $c$ (Lemma~\ref{lemma:EquisatisfiabilityFreeInstantiation}), 
(ii) $x_2$ and $y_2$ have been instantiated with the abstract real values
$5+\varepsilon$ and $-\infty$ (Definitions~\ref{definition:BaseInstantiationPoints},~\ref{definition:BaseVariableInstantiationPoints}, Lemma~\ref{lemma:EquisatisfiabilityBaseInstantiation}), and
(iii) $x_1$ and $y_1$ have been instantiated with $-\infty$. 
The instantiation does not need to consider the simple bounds $y_1<7$, $y_2\leq 2$, because it is sufficient to explore the reals either from
$-\infty$ upwards -- in this case upper bounds on real-valued variables can be ignored -- or from $+\infty$ downwards (ignoring lower bounds), as is similarly done in linear quantifier elimination~\cite{Loos1993}. Also instantiation does not need
to consider the value $5+\varepsilon$ for $x_1$ and $y_1$, motivated by the fact that the arguments of $R$ are not affected by the constraint $x_2 \neq 5$ (Example~\ref{example:ConnectedArgumentPositions}, Definition~\ref{definition:ConnectedArgumentPositions}).

All abstract values are represented by Skolem constants over the reals, together with defining axioms. 
For the example, we introduce the fresh Skolem constants $\alpha_{-\infty}$ to represent $-\infty$ (a ``sufficiently small'' value) and $\alpha_{5+\varepsilon}$ to represent $5+\varepsilon$ (a value ``slightly larger'' than $5$)
together with axioms expressing $\alpha_{-\infty} < 2 < 5 < \alpha_{5+\varepsilon}$.
Eventually, we obtain the ground clause set 
\begin{center}
$\begin{array}{r@{\;}r@{\;\;\rightarrow\;\;}l}
\alpha_{5+\varepsilon}\neq 5 	& \wedge\; R(\alpha_{-\infty}) & Q(c,\alpha_{5+\varepsilon})\\
\alpha_{-\infty}\neq 5 & \wedge\; R(\alpha_{-\infty}) & Q(c,\alpha_{-\infty})\\
\alpha_{-\infty}<7 \wedge \alpha_{5+\varepsilon}\leq 2 & & Q(c,\alpha_{5+\varepsilon}) \vee R(\alpha_{-\infty})\\
\alpha_{-\infty}<7 \wedge \alpha_{-\infty}\leq 2 & & Q(c,\alpha_{-\infty}) \vee R(\alpha_{-\infty})\\
\multicolumn{3}{c}{\alpha_{-\infty} < 2}\\
\multicolumn{3}{c}{5 < \alpha_{5+\varepsilon}}
\end{array}$
\end{center}
which has the model
$\alpha_{-\infty}^{\A} = 1$, $\alpha_{5+\varepsilon}^{\A} = 6$, $R^{\A} = \{1\}$, $Q^{\A} = \{(c,6), (c,1)\}$, for instance.

Using the result that every BSR clause set with simple bounds has a satisfiability-preserving ground instantiation
with respect to finitely many constants, we prove NEXP\-TIME-completeness of the fragment in Section~\ref{section:Complexity}.
The decidability result on simple bounds can be lifted to constraints of the form $x \triangleleft s$ where $x$ is the only
variable and $s$ a ground expression. The lifting is done by the introduction of further Skolem constants for complicated ground terms,
following ideas of~\cite{Kruglov2012} and presented in Section~\ref{section:Basification}. The paper ends with a conclusion, Section~\ref{sec:conclusion},
where we, in particular, discuss the impact of our results on automated reasoning procedures for the fragment and relevant areas of application.
The technical parts are moved to an appendix.


\section{Basic Definitions}

Hierarchic combinations of first-order logic with background theories build upon sorted logic with equality \cite{Bachmair1994b}. We instantiate this framework with the BSR fragment and linear arithmetic over the reals as the \emph{base theory}. The \emph{base sort $\R$} shall always be interpreted as the reals $\Real$. For simplicity, we restrict our considerations to a single \emph{free sort $\S$}, which may be freely interpreted as some nonempty domain, as usual. 

To build up atoms over the base theory, we introduce a countably infinite set $V_\R$ of \emph{variables of the base sort}, usually denoted $x, y, z$, two sets of constant symbols $\Omega_\LA := \{r \mid r\in\Real \} \cup \{c_1, \ldots, c_\kappa\}$ and $\Omega_\LA^\alpha  := \{\alpha_{-\infty}\} \cup \{\alpha_{c+\varepsilon} \,|\, c\in\Omega_\LA\}$, and the predicates typical of linear arithmetic $\Pi_\LA := \{ <,\; \leq,\; \equals,\; \not\equals,\; \geq,\; >\}$ with their standard meaning. While the semantics of the constants $r \in \Real$ and the predicates in $\Pi_\LA$ is predefined, the exact values assigned to constant symbols $c_1, \ldots, c_\kappa$ shall be determined by a formal interpretation (see below). These constant symbols can be conceived as \emph{existentially quantified} variables, and therefore we refer to $c_1, \ldots, c_\kappa$ as \emph{Skolem constants}. The exact number $\kappa$ of additional Skolem constants is left open; it may be $0$.
In fact, we introduce even more Skolem constants by adding the set $\Omega_\LA^\alpha$. For notational convenience, we syntactically distinguish this special kind of base-sort constant symbols $\alpha_\t$, $\t$ being of the form $-\infty$ or $d+\varepsilon$ for arbitrary base-sort constant symbols $d\in\Omega_\LA$. These will play a key role when we instantiate base-sort variables later on. We associate an inherent meaning with constant symbols $\alpha_\t$ that will be formalized by means of special axioms: $\alpha_{-\infty}$ shall stand for a value smaller than any values assigned to other occurring constant symbols, and $\alpha_{d+\varepsilon}$ shall be assigned a value that is a little larger the $d$'s value but not too large. 
\medskip

In order to hierarchically extend the base theory by the BSR fragment, we introduce the free sort $\S$, a countably-infinite set $V_\S$ of \emph{free-sort variables}, usually denoted $u, w$, a finite set $\Omega$ of \emph{free constant symbols of sort $\S$} and a finite set $\Pi$ of \emph{free predicate symbols} equipped with appropriate sort information. Note that every predicate symbol in $\Pi$ has a finite, nonnegative arity and can be of a mixed sorting over the two sorts $\R$ and $\S$. For instance $P : \R \times \S \times \R$ denotes a predicate with two real-valued arguments and one argument of the free sort.
We use the symbol $\approx$ to denote the built-in equality predicate on $\S$. To avoid confusion, we assume each constant or predicate symbol to occur in at most one of the sets $\Omega_\LA$, $\Omega_\LA^\alpha$, $\Omega$, $\Pi_\LA$, $\Pi$ and that none of these sets contains the $\approx$ symbol.

\begin{definition}[BSR Clause Fragment with Simple Bounds]\label{definition:BSRwithConstrSyntax}	
		Let $\Omega$ be a finite set of constant symbols $c$ of the free sort, and let $\Pi$ be a finite set of predicate symbols of finite arity over the sorts $\R$ and $\S$.
		
		An \emph{atomic constraint} is of the form $c\triangleleft d$ or $x\triangleleft d$ or $\alpha_\t \triangleleft d$ or $x\equals \alpha_\t$ or $\alpha_\t \equals \alpha_{\t'}$ with $c,d\in\Omega_\LA$, $x\in V_\R$, $\alpha_\t, \alpha_{\t'}\in\Omega_\LA^\alpha$ and $\triangleleft \in \{<,\leq,\equals,\not\equals,\geq,>\}$.
				
		A \emph{free atom} $A$ is either of the form $s\approx s'$ with $s,s'$ being free-sort constant symbols in $\Omega$ or free-sort variables in $V_\S$, respectively, or $A$ is of the form $P(s_1, \ldots, s_m)$, where $P:\xi_1\times\ldots\times\xi_m$ is an $m$-ary predicate symbol taken from $\Pi$. For each $i\leq m$ the term $s_i$ shall be of the sort $\xi_i$. If $\xi_i = \R$, then $s_i$ must be a variable $x\in V_\R$, and in case of $\xi_i = \S$, $s_i$ may be a variable $u\in V_\S$ or a constant symbol $c\in\Omega$.
		
		A \emph{clause} has the form $\Lambda \;\|\; \Gamma \to \Delta$, where $\Lambda$ is a multiset of atomic constraints, and $\Gamma$ and $\Delta$ are multisets of free atoms. We usually refer to $\Lambda$ as the \emph{constraint part} and to $\Gamma\to\Delta$ as the \emph{free part} of the clause.\\
		We conveniently denote the union of two multisets $\Theta$ and $\Theta'$ by juxtaposition $\Theta, \Theta'$. More-over, we often write $\Theta, A$ as an abbreviation for the multiset union $\Theta \cup \{A\}$. In our clause notation empty multisets are usually omitted left of ``$\to$'' and  denoted by $\Box$ right of ``$\to$'' (where $\Box$ at the same time stands for \emph{falsity}).
\end{definition}
The introduced clause notation separates linear arithmetic constraints from the free first-order part. We use the vertical double bar ``$\|$'' to indicate this on the syntactic level.  (On the semantic level ``$\|$'' is meaningless.) 
Intuitively, clauses $\Lambda \;\|\; \Gamma \to \Delta$ can be read as $\bigl(\bigwedge\Lambda \wedge \bigwedge\Gamma\bigr) \to \bigvee\Delta$, i.e.\ the multiset $\Lambda$ stands for a conjunction of atomic constraints, $\Gamma$ stands for a conjunction of the free atoms in it and $\Delta$ stands for a disjunction of the contained free atoms. All occurring variables are implicitly universally quantified. 

Requiring the free part $\Gamma\to\Delta$ of clauses to not contain any base-sort constant symbols does not limit expressiveness. Every base-sort constant symbol $c$ in the free part can safely be replaced by a fresh base-sort variable $x_c$ when an atomic constraint $x_c \equals c$ is added to the constraint part of the clause (a process known as purification \cite{Bachmair1994b, Kruglov2012}).

In the rest of the paper we omit the phrase ``over the BSR fragment with simple bounds'' when talking about clauses and clause sets, although we mainly restrict our considerations to this fragment.

A \emph{hierarchic interpretation} is an algebra $\A$ which interprets the base sort $\R$ as $\R^\A = \Real$, assigns real values to the Skolem constants in $\{c_1, \ldots, c_\kappa\}$ and $\Omega_\LA^\alpha$ and interprets all constants in $\Omega_\LA\cap\Real$ and predicates in $\Pi_\LA$ in the standard way. Moreover, $\A$ comprises a nonempty domain $\S^\A$, assigns to each free-sort constant symbol $c$ in $\Omega$ a domain element $c^\A \in \S^\A$ and interprets every predicate symbol $P\!:\!\xi_1\times\ldots\times\xi_m$ in $\Pi$ by a set $P^\A\subseteq \xi_1^\A\times\ldots\times\xi_m^\A$. Summing up, $\A$ extends the standard model of linear arithmetic and adopts the standard approach to semantics of (sorted) first-order logics when interpreting the free part of clauses.

Given a hierarchic interpretation $\A$, a \emph{variable assignment} is a sort-respecting total mapping $\beta: V_\R\cup V_\S \to \R^\A \cup \S^\A$. 
We write $\A(\beta)(s)$ to mean the \emph{value of the term $s$ under $\A$ with respect to the variable assignment $\beta$}. In accordance with the notation used so far, we thus define $\A(\beta)(v) := \beta(v)$ for every variable $v$ and $\A(\beta)(c) := c^\A$ for every constant symbol $c$.
As usual, we use the symbol $\models$ to denote truth under a hierarchic interpretation $\A$, possibly with respect to a variable assignment $\beta$.
In detail, we have the following for atomic constraints, equational and nonequational free atoms and clauses, respectively:
	\begin{itemize}
		\item $\A,\beta\models s\triangleleft s'$ if and only if $\A(\beta)(s) \triangleleft \A(\beta)(s')$,
		\item $\A,\beta\models s\approx s'$ if and only if $\A(\beta)(s) = \A(\beta)(s')$,
		\item $\A,\beta\models P(s_1, \ldots, s_m)$ if and only if $\bigl\<\A(\beta)(s_1), \ldots, \A(\beta)(s_m)\bigr\>\in P^\A$,
		\item $\A,\beta\models \Lambda\;\|\;\Gamma\to\Delta$ if and only if
			\begin{itemize}
				\item $\A,\beta\not\models s\triangleleft s'$ for some atomic constraint $(s\triangleleft s')\in\Lambda$, or
				\item $\A,\beta\not\models A$ for some free atom $A\in\Gamma$, or 
				\item $\A,\beta\models B$ for some free atom $B\in\Delta$. 
			\end{itemize}
	\end{itemize}
The variables occurring in clauses shall be universally quantified. Therefore, given a clause $C$, we call $\A$ a \emph{hierarchic model of $C$}, denoted $\A\models C$, if and only if $\A,\beta\models C$ holds for every variable assignment $\beta$. For clause sets $N$, we say $\A$ is a \emph{hierarchic model of $N$}, also denoted $\A\models N$, if and only if $\A$ is a hierarchic model of every clause in $N$.
We call a clause $C$ (a clause set $N$) \emph{satisfiable} if and only if there exists a hierarchic model $\A$ of $C$ (of $N$). 

\emph{Substitutions} $\sigma$ shall be defined in the standard way as sort-respecting mappings from variables to terms over our underlying signature. The \emph{restriction of the domain of a substitution $\sigma$ to a set $V$ of variables} is denoted by $\sigma|_V$ and shall be defined so that $v\sigma|_V := v\sigma$ for every $v\in V$ and $v\sigma|_V = v$ for every $v\not\in V$. While the application of a substitution $\sigma$ to terms, atoms and multisets thereof can be defined as usual, we need to be more specific for clauses. Consider a clause $C = \Lambda \;\|\; \Gamma \to \Delta$ and let $x_1, \ldots, x_k$ denote all base-sort variables occurring in $C$ for which $x_i\sigma \neq x_i$. We then set $C\sigma := \Lambda\sigma, x_1 \equals x_1\sigma,\, \ldots,\, x_k \equals x_k\sigma \;\|\; (\Gamma\to\Delta)\sigma|_{V_\S}$.
\emph{Simultaneous substitution} of $k \geq 1$ distinct variables $v_1, \ldots\, v_k$ with terms $s_1, \ldots, s_k$ shall be denoted in vector notation $\subst{v_1, \ldots, v_k}{s_1, \ldots, s_k}$ (or $\subst{\bar{v}}{\bar{s}}$ for short), where we require every $s_i$ to be of the same sort as $v_i$, of course.

Consider a clause $C$ and let $\sigma$ be an arbitrary substitution. The clause $C\sigma$ and variable-renamed variants thereof are called \emph{instances} of $C$.
A term $s$ is called \emph{ground}, if it does not contain any variables. 
A clause $C$ shall be called \emph{essentially ground} if it does not contain free-sort variables and for every base-sort variable $x$ occurring in $C$, there is an atomic constraint $x\equals d$ in $C$ for some constant symbol $d\in\Omega_\LA\cup\Omega_\LA^\alpha$. A clause set $N$ is \emph{essentially ground} if all the clauses it contains are essentially ground.

We can restrict the syntactic form of clauses even further without limiting expressiveness.
\begin{definition}[Normal Form of Clauses and Clause Sets]\label{definition:BSRwithConstrNormalform}
		A clause $\Lambda \;\|\; \Gamma \to \Delta$ is in \emph{normal form} if all base-sort variables which occur in $\Lambda$ also occur in $\Gamma \to \Delta$.
		A clause set $N$ is in \emph{normal form} if all clauses in $N$ are in normal form and pairwise variable disjoint.
\end{definition}
Let us briefly clarify why the above requirement on clauses does not limit expressiveness. 
Any base-sort variable $x$ not fulfilling the stated requirement can be removed from the clause $\Lambda\;\|\;\Gamma\to\Delta$ by existential quantifier elimination methods that transform $\Lambda$ into an equivalent constraint $\Lambda'$ in which $x$ does not occur.\footnote{Methods for the elimination of existentially quantified real variables include Fourier-Motzkin variable elimination \cite{Dantzig1973}, the Loos-Weispfenning procedure \cite{Loos1993} and many others, see e.g.\ Chapter 5 in \cite{Kroening2008}.} Moreover, $\Lambda'$ can be constructed in such a way that it contains only atomic constraints of the form admitted in Definition \ref{definition:BSRwithConstrSyntax} and so that no variables or constant symbols other than the ones in $\Lambda$ are necessary.

Given a clause set $N$, we will use the following notation: the set of all constant symbols occurring in $N$ shall be denoted by $\consts(N)$. While the set $\bconsts(N)$ exclusively contains all base-sort constant symbols from $\Omega_\LA$ that occur in $N$, all base-sort constant symbols $\alpha_\t$ appearing in $N$ shall be collected in the set $\aconsts(N) := \consts(N)\cap\Omega_\LA^\alpha$. The set of all free-sort constant symbols in $N$ is called $\fconsts(N)$. Altogether, the sets $\bconsts(N)$, $\aconsts(N)$ and $\fconsts(N)$ together form a partition of $\consts(N)$ for every clause set $N$.
We denote the set of all variables occurring in a clause $C$ (clause set $N$) by $\vars(C)$ ($\vars(N)$).


\section{Instantiation of Base-Sort Variables and Free-Sort Variables}\label{section:BaseModelTheory}

We first summarize the overall approach for base-sort variables in an intuitive way. To keep the informal exposition simple, we pretend that all base-sort constant symbols are taken from $\Real$. Consequently, we can speak of real values instead of constant symbols, and even refer to improper values such as $-\infty$ (a ``sufficiently small'' real) and $r+\varepsilon$ (a real ``slightly larger than $r$ but not too large''). A formal treatment with proper definitions will follow.

Given a finite clause set $N$, we intend to partition $\Real$ into finitely many parts such that satisfiability of $N$ necessarily leads to the existence of a hierarchic model $\A$ with the following property:
\begin{enumerate}[label=($\mathfrak{Pr}$), ref=($\mathfrak{Pr}$)]
	\item\label{enum:IntervalModel} Under $\A$ the free predicates occurring in $N$ cannot distinguish values that belong to the same part.
\end{enumerate}
Put differently, we have $\<\ldots, r, \ldots\> \in Q^\A$ if and only if $\<\ldots, r', \ldots\> \in Q^\A$, for every free predicate symbol $Q$ occurring in $N$, every part $p$ of $\Real$ and arbitrary reals $r, r' \in p$.
As soon as we found such a finite partition $\P$, we pick one real value $r_p \in p$ as representative from every part $p\in\P$. The following observation motivates why we can use those representatives instead of using universally quantified variables: Given a clause $C$ that contains a base-sort variable $x$, and given a set $\{d_1, \ldots, d_{k}\}$ of constant symbols such that $\{d_1^\A, \ldots, d_k^\A\} = \{r_p \mid p\in\P\}$ it holds ($*$) $\A\models C \;\Longleftrightarrow\; \A\models \bigl\{C\subst{x}{d_i} \bigm| 1\leq i \leq k \bigr\}$. The equivalence claims that we can transform universal quantification over the base domain into finite conjunction over all representatives of parts in $\P$. The formal version of this statement is given in Lemma \ref{lemma:EquisatisfiabilityBaseInstantiation}, and hierarchic models complying with property \ref{enum:IntervalModel} play a key role in its proof.
Since $\P$ is supposed to be finite, the resulting set of instances $\bigl\{C\subst{x}{d_i} \mid 1\leq i\leq k\bigr\}$ is finite, too.

It turns out that the \emph{elimination sets} described by Loos and Weispfenning  in \cite{Loos1993} (in the context of quantifier elimination for linear arithmetic) can be adapted to yield reasonable sets of representatives from which we can construct a finite partition exhibiting the described characteristics. In this approach the parts are intervals on the real axis.
	
For reasons of efficiency, we will operate on a more fine-grained level than described so far, as we define partitions of $\Real$ independently for certain groups of base-sort variables (induced by \emph{argument position classes}, cf.\ Definition \ref{definition:ConnectedArgumentPositions}). The possible benefit is a significant decrease in the number of necessary instances. But there is even more potential for savings. The complete line of definitions and arguments will be laid out in detail for one direction of instantiation along the real axis, namely in the positive direction starting from $-\infty$ (the ``default representative'') and going on to instantiation points $r+\varepsilon$ directly behind occurring constants $r$. However, one could as well proceed in the dual way, starting from $+\infty$ as the default and continuing with representatives $r-\varepsilon$ slightly smaller than occurring constants $r$. This duality is very promising from a practical point of view. As it turns out, one can choose the direction of instantiation independently for each base-sort variable that needs to be instantiated. Hence, one can always pick the direction that results in fewer instantiation points. Such a strategy might again considerably cut down the number of instances in the resulting ground clause set and therefore might lead to shorter processing times in automated reasoning procedures.

Formally, the aforementioned representatives are induced by constant symbols when interpreted under a hierarchic interpretation. Independently of any interpretation these constant symbols will serve as symbolic \emph{instantiation points}, i.e.\ they will be used to instantiate base-sort variables similar to the $d_i$ in equivalence ($*$).
We start off by defining the set of instantiation points that need to be considered for instantiation of a given base-sort variable. But first, we need a proper notion of what it means for a base-sort variable to be affected by a constraint. The following example shall illustrate the involved issues.
\begin{example}\label{example:ConnectedArgumentPositions}
	Consider the following clauses:
	\[
	\begin{array}{rclcll}
		x_1\not\equals -5 	&\| 	&		&\to	&T(x_1), 	&Q(x_1, x_2) ~, \\
		y_1 < 2,\; y_2 < 0	&\|	&		&\to	&		&Q(y_1, y_2) ~, \\
		z_1 \geq 6		&\|	&T(z_1)	&\to	&\Box	~.
	\end{array}
	\]	
	The variables $x_1$, $y_1$, $y_2$ and $z_1$ are affected by the constraints in which they occur explicitly. Technically, it is more suitable to speak of the \emph{argument position} $\<T,1\>$ instead of variables $x_1$ and $z_1$ that occur as the first argument of predicate symbol $T$ in the first and third clause. Speaking in such terms, argument position $\<T,1\>$ is directly affected by the constraints $x_1 \not\equals -5$ and $z_1 \geq 6$, argument position $\<Q,1\>$ is directly affected by $x_1\not\equals -5$ and $y_1 < 2$, and finally $\<Q,2\>$ is affected by $y_2 < 0$.
	As soon as we take logical consequences into account, the notion of ``being affected'' needs to be extended. The above clause set, for instance, logically entails the clause $x \geq 6 \;\|\; \to Q(x,y)$. Hence, although not \emph{directly} affected by the constraint $z_1 \geq 6$ in the clause set, the argument position $\<Q,1\>$ is still indirectly subject to this constraint. The source of this effect lies in the first clause as it establishes a connection between argument positions $\<T,1\>$ and $\<Q,1\>$ via the simultaneous occurrence of variable $x_1$ in both argument positions.
\end{example}
One lesson learned from the example is that argument positions can be connected by variable occurrences. Such links in a clause set $N$ shall be expressed by the relation $\con_N$.

\begin{definition}[Connections between Argument Positions and Argument Position Classes]\label{definition:ConnectedArgumentPositions}
	Let $N$ be a clause set in normal form.
	We define the relation $\con_N$ to be the smallest equivalence relation over pairs in $\Pi\times\Nat$ such that $\<Q,j\> \con_N \<P,i\>$ whenever there is a clause in $N$ containing free atoms $Q(\ldots,v,\ldots)$ and $P(\ldots,v,\ldots)$ in which the variable $v$ occurs at the $j$-th and $i$-th argument position, respectively. (Note that $Q=P$ or $j=i$ is possible.)
	
	The relation $\con_N$ induces equivalence classes $[\<P,i\>]_{\con_N}$ of argument positions in the usual way. To simplify notation a bit, we write $[\<P,i\>]$ instead of $[\<P,i\>]_{\con_N}$ when the set $N$ is clear from the context.
	
	Consider a variable $v$ that occurs at the $i$-th argument position of a free atom $P(\ldots, v, \ldots)$ in $N$. We denote the \emph{argument position class of $v$ in $N$} by $\APC_N(v)$, i.e.\ $\APC_N(v) := [\<P,i\>]$. 
	If $v$ is a free-sort variable that exclusively occurs in equations $v\approx s$ in $N$, we set $\APC_N(v) := \emptyset$.
\end{definition}

Next, we collect the instantiation points that are necessary to eliminate base-sort variables by means of finite instantiation. In order to do this economically, we rely on the same idea that also keeps elimination sets in \cite{Loos1993} comparatively small (see the discussion below). 
\begin{definition}[Instantiation Points for Base-Sort Argument Positions]\label{definition:BaseInstantiationPoints}
	Let $N$ be a clause set in normal form and let $P: \xi_1\times\ldots\times \xi_m$ be a free predicate symbol occurring in $N$. Let $J := \{i \;|\; \xi_i = \R, 1\leq i\leq m \}$ be the indices of $P$'s base-sort arguments. For every $i\in J$, we define $\I_{P,i,N}$ to be the smallest set fulfilling
	\begin{enumerate}[label=(\roman{*}), ref=(\roman{*})]
		\item $d\in\I_{P,i,N}$ if there exists a clause $C$ in $N$ containing an atom $P(\ldots, x, \ldots)$ in which $x$ occurs as the $i$-th argument and a constraint $x\equals d$ or $x \geq d$ with $d\in\Omega_\LA$ appears in $C$, and
		\item $\alpha_{d+\varepsilon}\in\I_{P,i,N}$ if there exists a clause $C$ in $N$ containing an atom of the form $P(\ldots, x, \ldots)$, in which $x$ is the $i$-th argument and a constraint of the form $x\not\equals d$ or $x > d$ or $x\equals \alpha_{d+\varepsilon}$ with $d\in\Omega_\LA$ appears in $C$.
	\end{enumerate}
\end{definition}
The most apparent peculiarity about this definition is that atomic constraints of the form $x < d$ and $x \leq d$ are completely ignored when collecting instantiation points for $x$'s argument position. First of all, this is one of the aspects that makes this definition interesting from the efficiency point of view, because the number of instances that we have to consider might decrease considerably in this way. To develop an intuitive understanding why it is enough to consider constraints $x\triangleleft d$ with $\triangleleft \in \{\equals, \not\equals, \geq, >\}$ when collecting instantiation points, the following example may help.

\begin{example}
	Consider two clauses $C = x > 2,\, x \leq 5 \;\| \to T(x)$ and $D = x < 0 \;\|\; T(x)\to\Box$. Recall that we are looking for a finite partition $\P$ of $\Real$ such that we can construct a hierarchic model $\A$ of the clause set $\{C, D\}$ that complies with \ref{enum:IntervalModel}, i.e.\ for every part $p\in \P$ and arbitrary real values $r_1, r_2 \in p$ it shall hold $r_1 \in T^\A$ if and only if $r_2 \in T^\A$. A natural candidate for $\P$ is $\{(-\infty,0), [0,2], (2,5], (5,+\infty) \}$ which takes every atomic constraint in $C$ and $D$ into account. Correspondingly, we find the candidate predicate $T^\A = (2,5]$. Consequently, $\A$ is a hierarchic model of $C$ and $D$ alike and it exhibits property \ref{enum:IntervalModel}.
	
	But there are other interesting possibilities, for instance, the more coarse-grained partition $\{(-\infty, 2], (2, +\infty)\}$ together with the predicate $T^\A = (2, +\infty)$.
	This latter candidate partition completely ignores the constraints $x < 0$ and $x \leq 5$ that constitute upper bounds on $x$ and in this way induces a simpler partition. Dually, we could have concentrated on the upper bounds instead (completely ignoring the lower bounds). This would have lead to the partition $\{(-\infty, 0), [0, 5], (5, +\infty)\}$ and candidate predicates $T^\A = [0, 5]$ (or $T^\A = [0, +\infty)$). Both ways are possible, but the former yields a coarser partition and is thus more attractive.
\end{example}
The example has revealed quite some freedom in choosing an appropriate partition of the reals. A large number of parts directly corresponds to a large number of instantiation points -- one for each interval --, and therefore leads a large number of instances that need to be considered by a reasoning procedure. Hence, with respect to efficiency, it is of high interest to keep the partition coarse.

\begin{definition}[Instantiation Points for Argument Position Classes and Induced Partition]\label{definition:BaseVariableInstantiationPoints}
	Let $N$ be a clause set in normal form and let $\A$ be a hierarchic interpretation.
	For every equivalence class $[\<P,i\>]$ induced by $\con_N$ we define the following:

	The set $\I_{[\<P,i\>], N}$ of \emph{instantiation points for $[\<P,i\>]$} is defined by $\I_{[\<P,i\>], N} := \{\alpha_{-\infty}\} \;\cup$ \linebreak $\bigcup_{\<Q,j\>\in[\<P,i\>]} \I_{Q,j,N}$.
	
	The sequence $r_1, \ldots, r_k$ shall comprise all real values in the set $\bigl\{c^\A \bigm| c\in\I_{[\<P,i\>], N}\cap\Omega_\LA \text{ or } \alpha_{c+\varepsilon}\in \I_{[\<P,i\>], N} \bigr\}$ ordered so that $r_1 < \ldots < r_k$.
	Given a real number $r$, we say \emph{$\I_{[\<P,i\>], N}$ $\A$-covers $r$} if there exists an instantiation point $c\in\I_{[\<P,i\>],N}\cap\Omega_\LA$ with $c^\A = r$; analogously $(r+\varepsilon)$ is \emph{$\A$-covered by $\I_{[\<P,i\>],N}$} if there is an instantiation point $\alpha_{c+\varepsilon}\in\I_{[\<P,i\>],N}$ with $c^\A = r$.

	The partition $\P_{[\<P,i\>],N,\A}$ of the reals into finitely many intervals shall be the smallest partition (smallest w.r.t.\ the number of parts) fulfilling the following requirements:
		
		\begin{enumerate}[label=(\roman{*}), ref=\roman{*}]
			\item\label{enum:BaseVariableInstantiationPoints:One} If $r_1$ is $\A$-covered by $\I_{[\<P,i\>],N}$, then $(-\infty, r_1)\in\P_{[\<P,i\>],N,\A}$; otherwise, $(-\infty, r_1]\in\P_{[\<P,i\>],N,\A}$.

			\item\label{enum:BaseVariableInstantiationPoints:Two} For every $j$, $1\leq j\leq k$, if $r_j$ and $r_{j}+\varepsilon$ are both $\A$-covered by $\I_{[\<P,i\>],N}$, then $[r_j, r_j] = \{r_j\}\in\P_{[\<P,i\>],N,\A}$.
								
			\item\label{enum:BaseVariableInstantiationPoints:Three} For every $\ell$, $1\leq \ell< k$,
				\begin{enumerate}[label=(\ref{enum:BaseVariableInstantiationPoints:Three}.\roman{*}), ref=(\ref{enum:BaseVariableInstantiationPoints:Three}.\roman{*})]
					\item\label{enum:BaseVariableInstantiationPoints:ThreeOne} if $r_\ell+\varepsilon$ and $r_{\ell+1}$ are both $\A$-covered by $\I_{[\<P,i\>],N}$, then $(r_\ell, r_{\ell+1})\in\P_{[\<P,i\>],N,\A}$,							
					\item\label{enum:BaseVariableInstantiationPoints:ThreeTwo} if $r_\ell+\varepsilon$ is $\A$-covered by $\I_{[\<P,i\>],N}$ but $r_{\ell+1}$ is not, then $(r_\ell, r_{\ell+1}]\in\P_{[\<P,i\>],N,\A}$,
					\item\label{enum:BaseVariableInstantiationPoints:ThreeThree} if $r_\ell+\varepsilon$ is not $\A$-covered by $\I_{[\<P,i\>],N}$ but $r_{\ell+1}$ is, then $[r_\ell, r_{\ell+1})\in\P_{[\<P,i\>],N,\A}$,							
					\item\label{enum:BaseVariableInstantiationPoints:ThreeFour} if neither $r_\ell+\varepsilon$ nor $r_{\ell+1}$ is $\A$-covered by $\I_{[\<P,i\>],N}$, then $[r_\ell, r_{\ell+1}]\in\P_{[\<P,i\>],N,\A}$.
				\end{enumerate}
					
			\item\label{enum:BaseVariableInstantiationPoints:Four} If $r_k+\varepsilon$ is $\A$-covered by $\I_{[\<P,i\>],N}$, then $(r_k, +\infty)\in\P_{[\<P,i\>],N,\A}$; otherwise $[r_k, +\infty)\in\P_{[\<P,i\>],N,\A}$.
		\end{enumerate}
\end{definition}
Please note that partitions as described in the definition do always exist, and do not contain empty parts.


We next fix the semantics of constant symbols $\alpha_\t$ by giving appropriate axioms that make precise what it means for $\alpha_{-\infty}$ to be ``small enough'' and for $\alpha_{c+\varepsilon}$ to be ``a little larger than $c$ but not too large''.

\begin{definition}[Axioms for Instantiation Points $\alpha_\t$]\label{definition:BaseInstantiationAxioms}
	Let $\I$ be a set of instantiation points and $\C\subseteq \Omega_\LA$ be a set of base-sort constant symbols. We define the \emph{set of axioms} $\Ax_{\I, \C}:= \bigl\{ \alpha_{-\infty} < c  \bigm| \alpha_{-\infty}\in\I \text{ and } c\in\C\bigr\} \;\cup$ $\bigl\{ d < \alpha_{d+\varepsilon} \bigm| \alpha_{d+\varepsilon}\in\I\} \;\cup\; \bigl\{ d < c \;\to\; \alpha_{d+\varepsilon} < c \bigm| \alpha_{d+\varepsilon}\in\I$ $\text{ and } c\in\C\setminus\{d\} \bigr\} \;\cup\; \bigl\{ d \equals c \;\to\; \alpha_{d+\varepsilon} \equals \alpha_{c+\varepsilon} \bigm| \alpha_{d+\varepsilon}\in\I$ $\text{ and } c\in\C\setminus\{d\} \bigr\}$.
\end{definition}
The axioms are clearly not in the admissible clause form. But they can easily be transformed into proper clauses with empty free parts. For convenience, we stick to the above notation, but keep in mind that formally the axioms have a form in accordance with Definition \ref{definition:BSRwithConstrSyntax}.

Due to the meaning that we have just formally given to instantiation points $\alpha_\t$, it can be proven (cf. Lemma \ref{lemma:PartitionRepresentatives} in the appendix) that we can always find a constant symbol $c_p$ in the set of instantiation points that we can use to represent a given interval $p \in \P$. 
We have arrived at one of the core results of the present paper. The next lemma shows that we can eliminate base-sort variables $x$ from clauses $C$ in a clause set $N$ by replacing $C$ with finitely many instances in which $x$ is substituted with the instantiation points that we computed for $x$. The resulting clause set shall be called $\widehat{N}_x$. In addition, the axioms that stipulate the meaning of newly introduced constant symbols $\alpha_\t$ need to be added to $\widehat{N}_x$. Iterating this step for every base-sort variable in $N$ eventually leads to a clause set that is essentially ground with respect to the constraint parts of the the clauses it contains (free-sort variables need to be treated separately, of course, see the end of this section).

\begin{lemma}\label{lemma:EquisatisfiabilityBaseInstantiation}
	Let $N$ be a clause set in normal form, and assume that $N$ contains all the axioms in $\Ax_{\aconsts(N), \bconsts(N)}$ and for every $\alpha_{e+\varepsilon} \in \aconsts(N)$ we have $e\in\bconsts(N)$.	
	Suppose there is a clause $C$ in $N$ which contains a base-sort variable $x$. Let the clause set $\widehat{N}_x$ be constructed as follows: $\widehat{N}_x := \bigl(N\setminus\{ C \}\bigr)  \;\cup\; \bigl\{C\subst{x}{c}	\;\bigm|\; c\in\I_{\APC_N(x), N} \bigr\} \;\cup\; \Ax_{\I_{\APC_N(x), N}, \bconsts(N)}$.
	The original clause set $N$ is satisfiable if and only if $\widehat{N}_x$ is satisfiable. 
\end{lemma}
Due to space limitations, we only provide an informal outline of the proof here. The full details are formulated in the appendix.
The line of argument leading from a hierarchic model of the original clause set $N$ to a hierarchic model of the modified clause set $\widehat{N}_x$ is almost trivial. The converse direction, however, rests on a model construction that yields a hierarchic model $\B$ which complies with property \ref{enum:IntervalModel}. Given a hierarchic model $\A$ of the instantiated clause set $\widehat{N}_x$, we construct the partition $\P_{\APC_N(x), N, \A}$ based on the argument position class associated to $x$ in the original clause set $N$. We know that each part $p$ in the partition can be represented by a specific instantiation point $c_p$ in $\I_{\APC_N(x), N}$ 
and for each of these instantiation points there is one clause in the modified clause set $\widehat{N}_x$ in which $x$ is instantiated by $c_p$.
Since $\B$ is supposed to comply with \ref{enum:IntervalModel}, i.e.\ it shall not distinguish between real values that belong to the same part, the information how the model $\A$ treats the representative of an interval can be transferred to all values from this interval. 
An example might be may help to illustrate the key ideas.
\begin{example}
	Consider the clause set $N = \bigl\{x > 2,\, z\equals 4 \;\|\; Q(x, z) \to T(x)\bigr\}$.	
	There are two instantiation points for variable $x$, namely $\alpha_{-\infty}$ and $\alpha_{2+\varepsilon}$. This leads to
		\begin{align*}
			\widehat{N}_x =\; &\bigl\{ \alpha_{-\infty} > 2,\, z\equals 4,\, x\equals \alpha_{-\infty} \;\|\; Q(x, z) \to T(x) \;,\\ 
						&\quad \alpha_{2+\varepsilon} > 2,\, z\equals 4,\, x\equals \alpha_{2+\varepsilon} \;\|\; Q(x, z) \to T(x) \bigr\}
						\;\;\cup\; \Ax_{\{\alpha_{-\infty}, \alpha_{2+\varepsilon}\}, \{2,4\}}
		\end{align*}
	where the axioms express the fact that the value of $\alpha_{-\infty}$ is strictly smaller than $2$ and the value of $\alpha_{2+\varepsilon}$ shall lie within the interval $(2, 4)$ in every hierarchic model of $\widehat{N}_x$. 
	
	We assume to have a hierarchic model $\A$ of $\widehat{N}_x$ at hand with $\alpha_{-\infty}^\A = 0$ and $\alpha_{2+\varepsilon}^\A = 3$. Then the two intervals $(-\infty, 2]$ and $(2,+\infty)$ constitute the partition $\P_{\APC_N(x), N, \A}$, and the partition $\P_{\APC_N(z), N, \A}$ is given by $\bigl\{(-\infty, 4), [4, +\infty)\bigr\}$. Moreover, assume that the set $Q^\A$ contains the pairs $\<0, 0\>$, $\<3,4\>$ and $\<7,4\>$ while the set $T^\A$ shall be the union of intervals $(0, 5] \cup [6,7)$. Clearly, $\A$ is a hierarchic model of $\widehat{N}_x$. (However, $\A$ is not a hierarchic model of $N$, since the pair $\<7,4\>$ is in $Q^\A$ but $7$ does not belong to $T^\A$.)
	
	We now construct a hierarchic model $\B$ from $\A$ such that $\B$ complies with \ref{enum:IntervalModel}.
	First of all, $\B$ takes over all values assigned to constant symbols, i.e.\ $\alpha_{-\infty}^\B := \alpha_{-\infty}^\A$ and $\alpha_{2+\varepsilon}^\B := \alpha_{2+\varepsilon}^\A$. The definition of $T$ under $\B$ will be piecewise with respect to the partition $\P_{[\<T,1\>], N, \A} = \P_{\APC_N(x), N, \A} = \{(-\infty, 2], (2, +\infty)\}$. To do so, we use the idea of representatives: the interval $(-\infty, 2]$ is represented by $\alpha_{-\infty}^\A = 0$ and the interval $(2, +\infty)$ has $\alpha_{2+\varepsilon}^\A = 3$ as its representative. At the same time, $\alpha_{-\infty}^\A$ also represents the interval $(-\infty, 4)$ in partition $\P_{[\<Q,2\>], N, \A}$, and $[4, +\infty)$ is represented by the constant $4$.
	Putting the abstract idea of property \ref{enum:IntervalModel} into action, the set $T^\B$ will be defined so that $(-\infty, 2] \subseteq T^\B$ if and only if $\alpha_{-\infty}^\A \in T^\A$, and $(2, +\infty) \subseteq T^\B$ if and only if $\alpha_{2+\varepsilon}^\A \in T^\A$. For the set $Q^\B$ matters are technically more involved but follow the same scheme. For instance, all pairs $\<r_1, r_2\>$ with $r_1 \in (2, +\infty)$ and $r_2 \in [4, +\infty)$ shall be in $Q^\B$ if and only if the pair $\<\alpha_{2+\varepsilon}^\A, 4\>$ is in $Q^\A$. 
	Consequently, we end up with $T^\B = (2, +\infty]$ and $Q^\B = \{ \<r_1, r_2\> \mid r_1 \in (-\infty, 2] \text{ and } r_2 \in (-\infty, 4) \} \cup \{\<r_1, r_2\> \mid r_1 \in (2, +\infty) \text{ and } r_2 \in [4, +\infty)\}$.
	The net results is a hierarchic interpretation $\B$ that is a model of $\widehat{N}_x$ just as $\A$ is. But beyond that $\B$ is also a hierarchic model of the original clause set $N$.
\end{example}

We now briefly turn our attention to the elimination of free-sort variables.
Regarding formulas of free first-order logic without equality, it is a well-known fact that every satisfiable formula has a Herbrand model. When considering finite clause sets over the Bernays-Sch\"onfinkel fragment, Herbrand universes simply comprise all occurring constant symbols. Consequently, exhaustive instantiation of variables by all occurring constant symbols leads to an equisatisfiable clause set that is ground and finite. This does not change, if we add equality to the fragment, and thus end up with BSR.
It turns out that we can directly transfer this idea to the BSR fragment extended with simple bounds. A detailed formulation is given in the appendix (Lemma \ref{lemma:EquisatisfiabilityFreeInstantiation}).


\section{The Time Complexity of Deciding Satisfiability}\label{section:Complexity}

In the last section we have seen how to eliminate base-sort and free-sort variables by means of finite instantiation. In this section, we put these instantiation mechanisms together to obtain a nondeterministic algorithm that decides the hierarchic satisfiability problem for the BSR fragment with simple bounds and investigate its complexity. As a measure of the length of clause sets, clauses, atoms and multisets thereof, we use the number of occurrences of constant symbols and variables in the respective object, and denote it by $\len(\cdot)$.
\begin{theorem}\label{theorem:Complexity}
	Let $N$ be a clause set (of length at least $2$) that does not contain any constant symbol from $\Omega_\LA^\alpha$ but at least one free-sort constant symbol. Satisfiability of $N$ can be decided in nondeterministic exponential time. To put it more precisely: the problem lies in $\text{NTIME}\bigl(\len(N)^{\mathfrak{c}\cdot\len(N)}\bigr)$ for some constant $\mathfrak{c} > 1$.
\end{theorem}
\begin{proof}[Proof Sketch]

	We devise a na\"ive algorithm that decides a given problem instance as follows.
	As input we assume a finite clause set $N$ such that $\aconsts(N) = \emptyset$ and $\fconsts(N)\neq\emptyset$.
	\begin{enumerate}[label=(\Roman{*}), ref=(\Roman{*})]
		\item\label{enum:DecisionAlgorithmShort:One} Transform the input clause set $N$ into normal form. 
		
		\item\label{enum:DecisionAlgorithmShort:Two} Instantiate all variables in $N$ to obtain an equisatisfiable finite set of ground clauses
			\begin{align*}
				 N' := \Bigl\{ C\subst{\overline{x}, \overline{u}}{\overline{c}, \overline{d}} \Bigm|\; 
				 	&C \in N \text{ and } \{x_1, \ldots, x_k\} = \vars(C) \cap V_\R \text{ and }\\
					 &\{u_1, \ldots, u_\ell\} = \vars(C) \cap V_\S \text{ and }\\
					 &c_i\in\I_{\APC_{N}(x_i), N} \text{ and } d_j\in\fconsts(N) \text{ for all $i$ and $j$} \Bigr\} ~.
			\end{align*}
			
			Let $\Ax$ be the union of all axiom sets $\Ax_{\{\alpha_\t\}, \bconsts(N)}$ such that $\alpha_\t$ occurs in $N'$.

		\item\label{enum:DecisionAlgorithmShort:Three} 
			Nondeterministically construct a hierarchic interpretation $\A$.		

		\item\label{enum:DecisionAlgorithmShort:Four} Check whether $\A$ is a hierarchic model of $N'\cup \Ax$.
	\end{enumerate}
	A detailed exposition of the algorithm and a proof of its correctness can be found in the appendix, alongside with a detailed analysis of its complexity.	
	At this point, we only take a brief look at the running time.
	
	In a worst-case scenario, step \ref{enum:DecisionAlgorithmShort:One} might lead to an exponential blowup in the length of the clause set. 
	Computing the information required to execute the instantiation in Step \ref{enum:DecisionAlgorithmShort:Two}, e.g.\ constructing the sets $\bconsts(N)$, $\fconsts(N)$, $\I_{\APC_N(x_i), N}$ and the relation $\con_N$ together with the induced equivalence classes, can all be done in polynomial time once for all base-sort variables. The same holds for collecting all necessary axioms in the set $\Ax$. Clearly, the instantiation itself might lead to another exponential blowup in the length of the clause set in the worst case. Finally, the nondeterministic construction of a candidate interpretation $\A$ and checking whether $N'\cup \Ax$ is valid under $\A$ consumes time polynomial in the clause set again.
\end{proof}

NEXPTIME-completeness of satisfiability for the Bernays-Sch\"onfinkel fragment (cf.\ \cite{Lewis1980}) immediately yields NEXPTIME-hardness of satisfiability for finite clause sets over the BSR fragment with simple bounds. Together with Theorem \ref{theorem:Complexity} we thus obtain NEXPTIME-completeness of the problem.
\begin{corollary}
	The problem of deciding satisfiability of finite clause sets over the BSR fragment with simple bounds is NEXPTIME-complete.
\end{corollary}


\section{Beyond Simple Bounds}\label{section:Basification}

In this section we long to answer the question how simple our constraints have to be. The most complex atomic constraints we have allowed by now are of the form $x \triangleleft c$. Being able to cope with this kind, we can leverage the idea of \emph{flattening} to deal with more complicated constraints such as $3x + c < 1$. 
The basic idea rests on two steps:
	\begin{enumerate}[label=(\Roman{*}), ref=(\Roman{*})]
		\item Transform this constraint into the equivalent $x < \tfrac{1}{3}-\tfrac{1}{3} c$.
		\item Introduce a fresh Skolem constant $b$, transform the constraint into $x < b$ and add a 
\emph{defining clause} $C_b$ expressing $b \equals \tfrac{1}{3}-\tfrac{1}{3} c$ to the clause set.
	\end{enumerate}
These two steps already indicate that this technique is restricted to atomic constraints that are either ground or univariate and linear, and in which the standard operations addition, subtraction, multiplication and division on the reals may be involved. But we may even allow free function symbols $g:\xi_1\times \ldots \times\xi_m \to \xi$ with $\xi_1, \ldots, \xi_m, \xi \in \{\R, \S\}$, as long as all subterms occurring below $g$ are ground. 

The key insight at this point is that complex ground terms $s$ can be replaced by fresh constant symbols $b_s$ of the corresponding sort, if we add a \emph{defining clause} $C_{b_s}$ that identifies $b_s$ with $s$ (a technique called \emph{basification} in \cite{Kruglov2012}) -- see below for details. Of course, the replacement of $s$ must be done consistently throughout the clause set. Consequently, we can extend the syntax of clauses so that more complex terms are admitted.

\begin{definition}[BSR Clause Fragment with Ground LA Bounds]\label{definition:extBSRwithConstrSyntax}
		Let $\Omega'$ collect nonconstant free function symbols equipped with sorting information.
		An \emph{atomic constraint} is of the form $s\triangleleft s'$ with $\triangleleft\in\{<, \leq, \equals, \not\equals, \geq, >\}$ and well-sorted base-sort terms $s, s'$ over symbols in $\Omega_\LA\cup$ $\{+, -, \cdot, /\} \cup \Omega \cup\Omega'$ in which every subterm $g(s_1, \ldots, s_m)$ with $g \in \Omega'$ is ground and where $s \triangleleft s'$ contains at most one occurrence of a base-sort variable and no free-sort variables.
		A \emph{free atom} $A$ is either of the form $s\approx s'$ with $s,s'$ being either free-sort variables or well-sorted ground terms of the free sort over function symbols in $\Omega_\LA \cup \{+,-,\cdot, /\} \cup \Omega \cup \Omega'$, or $A$ is of the form $P(s_1, \ldots, s_m)$, where $P:\xi_1\times\ldots\times\xi_m \in\Pi$ and for each $i\leq m$ the term $s_i$ is of sort $\xi_i$. If $\xi_i = \R$, then $s_i$ is a base-sort variable, and if $\xi_i = \S$, then $s_i$ is either a free-sort variable or a ground free-sort term over function symbols in $\Omega_\LA \cup \{+,-,\cdot, /\} \cup \Omega \cup \Omega'$.
		%
\end{definition}
As we have already sketched, we can easily transform a clause set $N$ that contains clauses over the extended syntax into an equisatisfiable set $N_1\cup N_2$ so that $N_1$ contains clauses according to Definition \ref{definition:BSRwithConstrSyntax} only and all clauses in $N_2$ are of the form $b_s \not\equals s \;\|\to\Box$ or $\|\to b_s \approx s$. Clearly, we can  proceed with $N_1$ according to Steps (\ref{enum:DecisionAlgorithm:One}) to (\ref{enum:DecisionAlgorithm:Five}) of the algorithm given in the proof of Theorem \ref{theorem:Complexity} and thus obtain the essentially ground clause set $N_1''$ that is equisatisfiable to $N_1$ and a set $\Ax$ of axioms. The construction of a hierarchic interpretation $\A$ for $N_1''$ in Step (\ref{enum:DecisionAlgorithm:Six}) can be modified so that it results in an interpretation that also covers all function symbols that occur in $N_2$ but not in $N_1$. Checking whether $\A$ is a hierarchic model of $N_1'' \cup \Ax \cup N_2$ can be done easily. Consequently, the hierarchic satisfiability problem for clause sets over the extension of the BSR fragment with ground LA bounds is decidable, too.


\section{Undecidable Fragments}\label{section:UndecidableFragments}

The existing results on undecidability of the BSR fragment extended with linear arithmetic 
either employ quantifier alternation~\cite{Halpern1991}, a complex first-order structure 
built out of several $n$-ary ($n>4$) predicates~\cite{CoxMT92}, or linear arithmetic
constraints with addition and subtraction~\cite{Fietzke2012}. 
We identify more fine-grained undecidable fragments on the basis of a
reduction of the halting problem for two-counter machines~\cite{Minsky1967} to 
the Bernays-Sch\"onfinkel-Horn (BSH) fragment extended with linear arithmetic.

In~\cite{Fietzke2012} the two-counter machine instructions are translated into clauses of the form
\begin{center}
$\begin{array}{rr@{\;\parallel\;}l@{\;\rightarrow\;}l}
(1) 	& x' = x + 1 		& M(b,x,y) & M(b',x',y)\\
(2) 	& x > 0, x' = x -1 	& M(b,x,y) & M(b',x',y)\\
 	& x = 0 & M(b,x,y) 	& M(b'',x,y)
\end{array}$
\end{center}
where $M(b,x,y)$ models the program at instruction label $b$ with counter values $x$, $y$.
Then, clause (1)~models the increment of counter $x$ (analogous for $y$) and a go to the next instruction $b'$;
clauses (2)~model the conditional decrement of $x$ (analogous for $y$) and, otherwise, a jump to instruction $b''$.
The start state is represented by a clause $\parallel \rightarrow M(\text{init},n,m)$ for two positive integer values $n$, $m$, and
the halt instruction is represented by an atom $M(\text{halt},x,y)$ and reachability of the halting state
by a clause $\parallel M(\text{halt},x,y)\rightarrow \Box$. Then, a two-counter machine program halts if and only if
the clause set constructed out of the program is unsatisfiable. Note that for this reduction it
does not matter whether integer or real arithmetic is the underlying arithmetic semantics.

As it turns out, two-counter machines can be encoded exclusively using very restricted syntax on the constraint part, 
such as difference constraints $x-y \triangleleft c$, additive constraints $x+y \triangleleft c$,  
quotient constraints $x \triangleleft c\cdot y$ (which could equivalently be written $\tfrac{x}{y} \triangleleft c$, hence the name) or 
multiplicative constraints $x \cdot y \triangleleft c$. Even more restrictive, in the case of difference constraints and quotient constraints 
only a single base-sort constant symbol is necessary. In case of quotient and multiplicative constraints, lower and upper bounds on the used variables do 
not lead to a decidable fragment -- which would be the case if we were using variables over the integers.

We start off with difference constraints. In our encoding $M(b, x, y, z)$ stands for a machine at instruction $b$ with counter 
values $i_1 = x - z - 1$ and $i_2 = y - z - 1$, where the last argument $z$ keeps track of an offset relative to which $x$ and $y$ store 
the values of the counters. Following this principle, the increment instruction for the first counter $i_1$ is encoded 
by the clause $x'-x\equals 1 \;\|\; M(b,x,y,z) \to M(b',x',y,z)$, which leaves the offset untouched. The offset is an 
appropriate tool that allows us to have a uniform syntactic structure for all atomic constraints. It is due to the offset 
encoding that we can easily use a difference constraint when checking whether a counter is zero or not. The conditional 
decrement instruction is split up in two clauses: the zero case $x-z\equals 1 \;\|\; M(b,x,y,z) \to M(b',x,y,z)$ and the 
non-zero case $x-z>1,\, y'-y\equals 1,\, z'-z\equals 1 \;\|\; M(b,x,y,z) \to M(b'',x',y,z')$. 
Hence, by undecidability of the halting problem for two-counter machines, we may conclude that satisfiability for the BSH 
fragment with difference constraints (requiring only the constant $1$ besides the input) and a single free $4$-ary predicate symbol is undecidable, too.

Encoding two-counter machines in the BSH fragment with quotient constraints works very similar. We only need to change the representation of 
counter values in a state $M(b, x, y, z)$ as follows: $i_1 = -\log_2(\tfrac{2x}{z}) = -\log_2(x) + \log_2(z) - 1$ and $i_2 = -\log_2(\tfrac{2y}{z})$. 
Incrementing the first counter is encoded by $2\cdot x'\equals x \;\|\; M(b,x,y,z) \to M(b',x',y,z)$, and the conditional decrement instruction is 
represented by $2\cdot x\equals z \;\|\; M(b,x,y,z) \to M(b',x,y,z)$ and $2\cdot x<z,\, 2\cdot y'\equals y,\, 2\cdot z'\equals z \;\|\; M(b,x,y,z) \to M(b'',x,y',z')$. 
Analogous to the case of difference constraints, we may now show undecidability of the satisfiability problem for the BSH fragment with 
quotient constraints (requiring only the constant $2$) and a single free $4$-ary predicate symbol is undecidable.
We have chosen negative exponents for the encoding of the counter values, since this guarantees that the range of the base-sort variables is bounded 
from below and above. Thus, we could restrict all base-sort variables to values within $(0, 1]$ (by adding appropriate atomic constraints to every clause), 
and still end up with an undecidable satisfiability problem. 

Having additive constraints of the form $x + y \triangleleft c$ at hand, we can simulate subtraction by defining the additive inverse using a 
constraint $x + x_- \equals  0$. To keep track of inverses, we adjust the arity of $M$ accordingly. Counter values are represented in the same way as we 
have done for difference constraints. The increment instruction for the first counter is thus encoded 
by $x' + x_- \equals 1,\, x' + x'_- \equals 0 \;\|\; M(b, x, x_-, y, y_-, z, z_-) \to M(b', x', x'_-, y, y_-, z, z_-)$. 
It is now straightforward to come up with the encoding of the conditional decrement. 		
Hence, satisfiability of the BSH fragment with additive constraints and a single free predicate symbol of arity $7$ is undecidable. 
However, this time we need two constants, namely $1$ and $0$.

In order to complete the picture, we shortly leave the realm of linear arithmetic and consider multiplicative constraints of the form $x \cdot y \triangleleft c$. 
These relate to quotient constraints as additive constraints relate to difference constraints. Hence, combining the previously used ideas of offsets and inverses, 
we can encode two-counter machines also with multiplicative constraints: 
$x \cdot x'_{-1} \equals  2,\, x' \cdot x'_{-1} \equals  1 \;\|\; M(b, x, x_{-1}, y, y_{-1}, z, z_{-1}) \to M(b', x', x'_{-1}, y, y_{-1}, z, z_{-1})$ 
encodes the increment instruction on the first counter, for instance. As in the case of quotient constraints, we could restrict the range of base-sort 
variables to $(0,1]$ by suitable constraints. Consequently, this leads to another fragment of which the satisfiability problem is undecidable.


\section{Conclusion and Future Work} \label{sec:conclusion}

Our main contributions are a proof showing that satisfiability of the BSR fragment equipped with simple bounds on real arguments is decidable
and NEXPTIME-complete, and that slight extensions to the simple bounds constraints yield undecidability. 
The complexity result is of particular interest, 
since satisfiability of the Bernays-Sch\"onfinkel fragment has already been known to be NEXPTIME-complete for more than three decades~\cite{Lewis1980}.

For the pure variable first-order
Horn fragment extended with simple bounds, DEXP\-TIME-completeness was shown in the logic programming framework~\cite{CoxMT92}.
The result relies on a dedicated, finite representation of the hull generated by the PROLOG model operator~\cite{Cox1993}. Our result relies
on a fine-grained identification of a finite set of ground instances that are sufficient to test satisfiability.
Our result can be directly used to turn instance-based automated reasoning calculi into efficient decision
procedures for the BSR (or BS) fragment extended with simple bounds~\cite{NieuwenhuisEtAl06, BaumgartnerFT08,Ruemmer08,BonacinaEtAl11,Kruglov2012}.
For example, for an SMT approach, the overall arithmetic reasoning can be encoded in polynomially
many propositional clauses and used for preprocessing. For the example at the end of Section~\ref{section:Introduction}, the
clause $\alpha_{-\infty} < 2$ reduces all atoms $\alpha_{-\infty}<7$, $\alpha_{-\infty}\leq 2$, $\alpha_{-\infty}\neq 5$ to true and this can even directly
be build into an instantiation process.

Finite domain fragments such as the BSR fragment equipped with simple bounds have many real-world applications.
Firstly, they enable up to exponentially more compact representations compared to ground first-order, SMT, or propositional formulae, e.g.,~\cite{EmmerKKV10}.
Secondly, they result naturally out of rich ontology languages such as OWL (\url{www.w3.org}) or YAGO(2,3)~\cite{SuchanekKW07} 
where simple bounds may represent comparisons
to time points or comparisons on finite set cardinalities. When modeling technical systems, guards in the form
of simple bounds are an inherent part for describing variance of the system, e.g.,  \cite{DhunganaEtAl13} or its behavior, e.g., \cite{AlurDill94,AlurCHHHNOSY95,BiereCCSZ03}.

The detailed analysis of necessary  ground instances carried out for arithmetic constraints in Section~\ref{section:BaseModelTheory}
can also be done for the free first-order part. It can further reduce the number of ground instances generated out of free-sort constants.
However, we did not do so for the sake of simplicity and brevity.

It turned out that we leave the realm of decidability as soon as we add elementary operations $\circ$ on the reals to our constraint language, 
even if we restrict the free part to Horn clauses over the Bernays-Sch\"onfinkel fragment. Constraints of the form $x\circ y \triangleleft c$ are 
sufficient to obtain an undecidable fragment, where $\circ$ can stand for addition, subtraction, multiplication and division, $\triangleleft$ represents 
standard relations $<, \leq, \equals, \not\equals, \geq, >$, and $c$ is a real-valued constant.
This observation nicely complements our main result, since it quite clearly highlights the limits of decidability in this context.
Moreover, it reveals some interplay between real-sorted constraints and the free first-order part. For instance, difference logic (boolean combinations of propositional 
variables and existentially quantified constraints $x - y \triangleleft c$ with $c$ being a real-valued constant and $\triangleleft \in \{<, \leq\}$) is known 
to be decidable~\cite{Mahfoudh2002}. However, we have seen in Section~\ref{section:UndecidableFragments} that its combination with the 
Bernays-Sch\"onfinkel-Horn fragment is sufficient to formalize two-counter machines.

Although the obvious options for further extending the constraint language lead to undecidability there might still be room for improvement. We leave it as future 
work to investigate the BSR fragment with simple bounds plus constraints of the form $x \triangleleft y$ with $x, y$ being real-valued variables 
and $\triangleleft \in \{<, \leq, \equals, \not\equals, \geq, >\}$. On the other hand, it is conceivable to combine other decidable free first-order 
fragments with simple bounds, preferably ones satisfying the finite model property such as the monadic fragment.
As we have already pointed out, a natural next step for us will be to devise useful decision procedures for the BSR 
fragment with simple bounds that perform well in practice.


\noindent
\paragraph{Acknowledgments:}
	This work was partially supported by the German Transregional Collaborative Research Center \mbox{SFB/TR 14 AVACS}.


\newpage
\appendix

\section*{Appendix}

\subsection*{Details related to instantiation of Base-Sort Variables (Section \ref{section:BaseModelTheory})}

\begin{proposition}\label{proposition:ImproperAxiomSatisfiability}
	Let $\I$ be a nonempty set of instantiation points and let $\C$ be a nonempty set of base-sort constant symbols from $\Omega_\LA$. Given an arbitrary hierarchic interpretation $\A$, we can construct a hierarchic model $\B$ so that (ii) $\B\models \Ax_{\I, \C}$ and $\B$ and $\A$ differ only in the interpretation of the constant symbols in $\aconsts(\Ax_{\I, \C})$.
\end{proposition}
\begin{proof}
	 The model $\B$ can be derived from $\A$ by redefining the interpretation of all constant symbols $\alpha_\t \in \I$ as follows:
	$\alpha_{-\infty}^\B := \min\{c^\A \mid c\in\C\} - 1$ and $\alpha_{d+\varepsilon}^\B := \tfrac{1}{2} \bigl( d^\A + \min(\{c^\A \mid c\in\C,\; c^\A > d^\A\}\cup \{d^\A+1\}) \bigr)$.
	 Everything else is taken over from $\A$.
\end{proof}

In order to have a compact notation at hand when doing case distinctions on intervals, we use ``$\ldbrace$''  to stand for ``$($'' as well as for ``$[$''. Analogously, ``$\rdbrace $'' is used to address the two cases ``$)$'' and ``$]$'' at the same time.
\begin{lemma}\label{lemma:PartitionRepresentatives}
	Let $N$ be a clause set in normal form and let $\A$ be some hierarchic model of $\Ax_{\I_{[\<P,i\>],N}, \bconsts(N)}$ for an arbitrary argument position pair $\<P,i\>$.
	For every part $p\in \P_{[\<P,i\>], N, \A}$ we find some constant symbol $c_{[\<P,i\>], p} \in \I_{[\<P,i\>], N}$ so that $c_{[\<P,i\>], p}^\A \in p$ and
	\begin{enumerate}[label=(\roman{*}), ref=(\roman{*})]
		\item\label{enum:PartitionRepresentatives:One} if $p = (-\infty, r_u\rdbrace $ or $p = (-\infty, +\infty)$, then $c_{[\<P,i\>], p}^\A < d^\A$ for every $d\in\bconsts(N)$,
		\item\label{enum:PartitionRepresentatives:Two} if $p = [r_\ell, r_u\rdbrace $ or $p = [r_\ell, +\infty)$, then $c_{[\<P,i\>], p}^\A = r_\ell$.
		\item\label{enum:PartitionRepresentatives:Three} if $p = (r_\ell, r_u\rdbrace $ or $p = (r_\ell, +\infty)$, then $r_\ell < c_{[\<P,i\>], p}^\A$ and $c_{[\<P,i\>], p}^\A < d^\A$ for every $d\in\bconsts(N)$ with $r_\ell < d^\A$.
	\end{enumerate}
\end{lemma}
\begin{proof}
	We proceed by case distinction on the form of the interval $p$.
	\begin{description}
		\item Case $p = (-\infty, r_u\rdbrace $ for some real value $r_u$. 
			By construction of $\P_{[\<P,i\>], N, \A}$, there is a constant symbol $e\in\bconsts(N)$ so that $e^\A = r_u$.		
			Moreover, we find the instantiation point $\alpha_{-\infty}$ in $\I_{[\<P,i\>], N}$ and thus $\Ax_{\I_{[\<P,i\>],N}, \bconsts(N)}$ contains the axiom $\alpha_{-\infty} < d$ for every constant symbol $d\in\bconsts(N)$, in particular $\alpha_{-\infty} < e$. Hence, $\alpha_{-\infty}^\A \in (-\infty, r_u)$.

		\item Case $p = (-\infty, +\infty)$. 
			We find the instantiation point $\alpha_{-\infty}$ in $\I_{[\<P,i\>], N}$, and obviously $\alpha_{-\infty}^\A$ lies in $p$. Moreover, $\Ax_{\I_{[\<P,i\>],N}, \bconsts(N)}$ contains the axiom $\alpha_{-\infty} < d$ for every constant symbol $d\in\bconsts(N)$.

		\item Case $p = [r_\ell, r_u\rdbrace $ for real values $r_\ell, r_u$ with $r_\ell \leq r_u$.
			\begin{description}
				\item If $p$ is a point interval $[r_\ell, r_\ell]$, then we find a constant symbol $e\in\I_{[\<P,i\>], N}\cap\Omega_\LA$ so that $p = \{e^\A\}$.
			
				\item If $p$ is not a point interval, i.e.\ $r_\ell < r_u$, then the definition of $\P_{[\<P,i\>], N, \A}$ entails the existence of a constant symbol $e \in \bconsts(N)$ such that $e^\A = r_\ell$ and either $e \in \I_{[\<P,i\>], N}$ or $\alpha_{e+\varepsilon}\in\I_{[\<P,i\>], N}$. But since $p$ is not a point interval, $r_\ell + \varepsilon$ cannot be $\A$-covered by $\I_{[\<P,i\>],N}$, and thus $\alpha_{e+\varepsilon}$ cannot be in $\I_{[\<P,i\>],N}$. Consequently, $e$ must be in $\I_{[\<P,i\>],N}$ and we can choose $c_{[\<P,i\>], p} := e$.
			\end{description}
			
		\item The case of $p = [r_\ell, +\infty)$ can be argued analogously to the previous case.

		\item Case $p = (r_\ell, r_u\rdbrace $ for real values $r_\ell, r_u$ with $r_\ell < r_u$.
		
			By construction of $\P_{[\<P,i\>], N, \A}$, we find some instantiation point $\alpha_{e+\varepsilon}\in \I_{[\<P,i\>], N}$ so that $e^\A = r_\ell$. Moreover, there exists a second constant symbol $e'\in\bconsts(N)$ for which $e'^\A = r_u$.			
			As a consequence of $\alpha_{e+\varepsilon}\in \I_{[\<P,i\>], N}$ the set $\Ax_{\I_{[\<P,i\>],N}, \bconsts(N)}$ contains the axioms $e < \alpha_{e+\varepsilon}$ and $e<d \to \alpha_{e+\varepsilon} < d$ for every constant symbol $d\in\bconsts(N)$, in particular for $d=e'$. And since $\A$ is a model of $\Ax_{\I_{[\<P,i\>],N}, \bconsts(N)}$, requirement \ref{enum:PartitionRepresentatives:Three} is satisfied and furthermore $\alpha_{e+\varepsilon}^\A\in p$ follows, since $e^\A < \alpha_{e+\varepsilon}^\A < e'^\A$.
			
		\item Case $p = (r_\ell, +\infty)$. There must be an instantiation point $\alpha_{e+\varepsilon}$ in $\I_{[\<P,i\>],N}$ such that $e^\A = r_\ell$.
			As we have already done in Definition \ref{definition:BaseVariableInstantiationPoints}, we denote by $r_1, \ldots, r_k$ all real values in ascending order which $\A$ assigns to constant symbols $d\in\bconsts(N)$ that either themselves are instantiation points in $\I_{[\<P,i\>],N}$ or for which $\alpha_{d+\varepsilon}$ occurs in $\I_{[\<P,i\>],N}$.	 
			By construction of $\P_{[\<P,i\>], N, \A}$, we know $r_\ell = r_k >  r_{k-1} > \ldots > r_1$, and thus there is no base-sort constant symbol $e'$ in $\bconsts(N)$ fulfilling $r_\ell < e'^\A$. Consequently, \ref{enum:PartitionRepresentatives:Three} is satisfied if we choose $c_{[\<P,i\>], p} := \alpha_{e+\varepsilon}$. Finally, we may conclude $\alpha_{e+\varepsilon}^\A \in p$, because $\Ax_{\I_{[\<P,i\>],N}, \bconsts(N)}$ contains the axiom $e < \alpha_{e+\varepsilon}$.
			\qedhere
	\end{description}
\end{proof}

\subsubsection*{Proof of Lemma \ref{lemma:EquisatisfiabilityBaseInstantiation}}
\begin{proof}
	We will reuse the notation ``$\ldbrace$'' and ``$\rdbrace$'' that we have defined right before Lemma \ref{lemma:PartitionRepresentatives}.

	The ``only if''-part is almost trivial to show:
			Let $\A$ be a hierarchic model of $N$. 			
			Every axiom in $\Ax_{\I_{\APC_N(x), N}, \bconsts(N)}$ which does not appear in $N$ concerns constant symbols $\alpha_\t$ that do not occur in $N$. By virtue of Proposition \ref{proposition:ImproperAxiomSatisfiability}, we can derive a hierarchic model $\B$ from $\A$ that differs from $\A$ only in how these $\alpha_\t$ are interpreted and that is a hierarchic model of $\Ax_{\I_{\APC_N(x), N}, \bconsts(N)} \setminus N$. Thus, $\B\models N$ follows from $\A\models N$.			
			Since the variable $x$ is universally quantified in $C$, $\B\models C$ entails $\B\models C\subst{x}{c}$ for every instance $C\subst{x}{c}$.			
			Altogether, we obtain $\B\models \widehat{N}_x$.
			
	The ``if''-part requires a more sophisticated argument:
		Let $\A$ be a hierarchic model of $\widehat{N}_x$.
		We use $\A$ to construct the hierarchic model $\B$ as follows.
		For the domain $\S^\B$ we reuse $\A$'s free domain $\S^\A$. 
		For all base-sort and free-sort constant symbols $c\in \consts(N)$, we set $c^{\B} := c^{\A}$.
		For every predicate symbol $P:\xi_1\times\ldots\times\xi_m \in \Pi$ that occurs in $N$ and for every argument position $i$, $1\leq i\leq m$, Lemma \ref{lemma:PartitionRepresentatives} guarantees the existence of a base-sort constant symbol $c_{[\<P,i\>], p}\in \I_{\APC_N(x), N}$ for every interval $p\in \P_{[\<P,i\>], N, \A}$, such that $c_{[\<P,i\>], p}^\A \in p$ and the requirements \ref{enum:PartitionRepresentatives:One} to \ref{enum:PartitionRepresentatives:Three} of the lemma are met.
		Based on this, we define the family of functions $\varphi_{[\<P,i\>]} : \Real\cup\S^\B \to \Real\cup\S^\A$ by
					\[ \varphi_{[\<P,i\>]}(a) :=
						\begin{cases}
							c_{[\<P,i\>], p}^\A		&\text{if $\xi_i = \R$ and $p\in \P_{[\<P,i\>], N, \A}$}\\
											&\text{is the interval $a$ lies in,} \\
							a				&\text{if $\xi_i = \S$.}
						\end{cases}
					\]
		Using the functions $\varphi_{[\<P,i\>]}$, we define the $P^\B$ so that for all domain elements $a_1,\ldots,a_m$ of appropriate sorts $\bigl\<a_1, \ldots, a_m\bigr\> \in P^{\B}$ if and only if $\bigl\<\varphi_{[\<P,1\>]}(a_1), \ldots,$ $\varphi_{[\<P,m\>]}(a_m)\bigr\> \in P^\A$.

		We next show $\B\models N$. Consider any clause $C' = \Lambda' \;\|\; \Gamma' \to \Delta'$ in $N$ and let $\beta : V_\R\cup V_\S\to \Real\cup \S^\B$ be an arbitrary variable assignment. 
		From $\beta$ we derive a special variable assignment $\beta_\varphi$ for which we will infer $\A,\beta_\varphi \models C'$ as an intermediate step:
			\[ \beta_\varphi(v) := 
					\begin{cases}
						c_{\APC_N(v), p}^\A	&\text{if $v\in V_\R$ and $\beta(v)\in p \in \P_{\APC_N(v), N, \A}$,}\\
						\beta(v)			&\text{if $v\in V_\S$,}
					\end{cases}
			\]
		for every variable $v$.
		If $C' \neq C$, then $\widehat{N}_x$ already contains $C'$, and thus $\A,\beta_\varphi\models C'$ must hold.
		In case of $C' = C$, let $p_*$ be the interval in $\P_{\APC_N(x), N, \A}$ containing the value $\beta(x)$, and let $c_*$ be an abbreviation for $c_{\APC_N(x), p_*}$. Due to $\beta_\varphi(x) = c_*^\A$ and since $\A$ is a model of the clause $C\subst{x}{c_*}$ in $\widehat{N}_x$, we conclude $\A,\beta_\varphi \models C$. 
		Hence, in any case we can deduce $\A,\beta_\varphi\models C'$. By case distinction on why $\A,\beta_\varphi \models C'$ holds, we may transfer this result to obtain $\B,\beta\models C'$, too.
		\begin{description}
			\item Case $\A, \beta_\varphi \not\models c\triangleleft d$ for some atomic constraint $c\triangleleft d$ in $\Lambda'$ with base-sort constant symbols $c, d \in \Omega_\LA\cup\Omega_\LA^\alpha$ and $\triangleleft \in \{<, \leq, \equals, \not\equals, \geq, >\}$. Since $\B$ and $\A$ interpret constant symbols in the same way and independently  of a variable assignment, we immediately get $\B, \beta\not\models c\triangleleft d$.
			
			\item Case $\A,\beta_\varphi \not\models (y\triangleleft d) \in \Lambda'$ for an arbitrary base-sort variable $y$ and a constant symbol $d\in\Omega_\LA$. This translates to $\beta_\varphi(y) {\not\!\triangleleft}\, d^\A$. Let $p\in\P_{\APC_N(y), N, \A}$ be the interval which contains $\beta(y)$ and therefore also $\beta_\varphi(y)$. 
				\begin{description}
					\item If $d^\A$ lies outside of $p$, then $\beta_\varphi(y) \triangleleft d^\A$ if and only if $\beta(y) \triangleleft d^\A$, since $\beta_\varphi(y) \in p$. Thus, $d^\B = d^\A$ entails $\B, \beta \not\models y\triangleleft d$.
					
					\item If $p$ is the point interval $p = \{d^\A\}$, then $\beta(y) = \beta_\varphi(y)$, and thus $\B,\beta \not\models y\triangleleft d$.
					
					\item If $p = \ldbrace r_\ell, r_u\rdbrace $ and $r_\ell < d^\A \leq r_u$, then $\triangleleft \not\in \{<, \leq, \not\equals\}$, since $\beta_\varphi(y) = c_{\APC_N(y), p}^\A < d^\A$ (by \ref{enum:PartitionRepresentatives:Two} and \ref{enum:PartitionRepresentatives:Three} of Lemma \ref{lemma:PartitionRepresentatives}). Moreover, we conclude $d\not\in \I_{\APC_N(y), N}$, since otherwise $p$ would be of the form $p = [d^\A, r_u\rdbrace $ by the construction of $\P_{\APC_N(y),N,\A}$ (requirements (\ref{enum:BaseVariableInstantiationPoints:Two}), \ref{enum:BaseVariableInstantiationPoints:ThreeThree} and \ref{enum:BaseVariableInstantiationPoints:ThreeFour}). Therefore, $\triangleleft \not\in \{\equals, \geq\}$, since otherwise the instantiation point $d$ would be in $\I_{\APC_N(y), N}$. This only leaves $\triangleleft = {>}$. Hence, the constraint $y > d$ occurs in $N$ and thus we find the instantiation point $\alpha_{d+\varepsilon}$ in $\I_{\APC_N(y), N}$. Consequently, requirements \ref{enum:BaseVariableInstantiationPoints:ThreeTwo} and \ref{enum:BaseVariableInstantiationPoints:ThreeFour} of the construction of $\P_{\APC_N(y),N,\A}$ entail $p = \ldbrace r_\ell, d^\A]$, which leads to $\beta(y) \leq d^\A$ because of $\beta(y)\in p$. Therefore, $\B,\beta \not\models y > d$ must be true.
					
						The cases $p = \ldbrace r_\ell, +\infty)$, $p = (-\infty, r_u\rdbrace $ and $(-\infty, +\infty)$ with $r_\ell < d^\A \leq r_u$ can be handled by similar arguments.

					\item If $p = [d^\A, r_u\rdbrace $ and $d^\A < r_u$, then $\beta_\varphi(y) = c_{\APC_N(y), p}^\A = d^\A$ by \ref{enum:PartitionRepresentatives:Two} of Lemma \ref{lemma:PartitionRepresentatives}. Consequently, \mbox{$\triangleleft \not\in \{\leq, \equals, \geq\}$}. We conclude $\alpha_{d+\varepsilon} \not\in \I_{\APC_N(y), N}$, because otherwise $[d^\A, r_u\rdbrace $ would be a point interval, contradicting $d^\A < r_u$. Hence, the only remaining possibility is $\triangleleft = {<}$. But by $\beta(y)\in p$ we deduce $\beta(y) \geq d^\A$. Therefore, we clearly get $\B,\beta \not\models y < d$.
					
						The case $p = [d^\A, +\infty)$ is covered by analogous arguments.
				\end{description}
				
			\item Case $\A, \beta_\varphi \not\models (y\equals \alpha_\t) \in \Lambda'$ for an arbitrary variable $y$ and a constant symbol $\alpha_\t\in\Omega_\LA^\alpha$.
				Let $p\in\P_{\APC_N(y), N, \A}$ be the interval which contains $\beta(y)$. 
				We immediately conclude $\alpha_\t^\A \neq c_{\APC_N(y),p}^\A$ since $c_{\APC_N(y),p}^\A = \beta_\varphi(y) \neq \alpha_\t^\A$. 
				\begin{description}
					\item If $\alpha_\t^\A$ does not lie in $p$, then $\beta(y) \neq \alpha_\t^\A$, and thus $\B,\beta \not\models y\equals \alpha_\t$.
					\item Assume $\alpha_\t^\A$ lies in $p$. 
						By $c_{\APC_N(y),p}^\A \in p$ and the facts $\alpha_\t^\A \neq c_{\APC_N(y),p}^\A$ and $\alpha_\t^\A \in p$, $p$ cannot be a point interval.						
						Moreover, $p$ cannot be of the form $(-\infty, r_u\rdbrace $, since then $c_{\APC_N(y),p} = \alpha_{-\infty}$ and $r_u < \alpha_{c+\varepsilon}^\A$ for all $\alpha_{c+\varepsilon} \in \aconsts(N)$ would follow by construction of $\P_{\APC_N(y),N,\A}$ (since $\A\models \Ax_{\aconsts(N), \bconsts(N)}$), contradicting either $\alpha_\t \neq c_{\APC_N(y),p}$ or $\alpha_\t^\A \in p$.							
						   Consequently, there exists a constant symbol $d\in\bconsts(N)$ so that $\alpha_\t = \alpha_{d+\varepsilon}$. Hence, the atomic constraint $y\equals \alpha_{d+\varepsilon}$ contributes the instantiation point $\alpha_{d+\varepsilon}$ to $\I_{\APC_N(y),N}$.
						
						The constant symbol $c_{\APC_N(y),p}$ cannot stem from $\Omega_\LA$, since then $p$ would be of the form $[c_{\APC_N(y),p}^\A, +\infty)$ or $[c_{\APC_N(y),p}^\A, r_u\rdbrace $ for some real value $r_u$ and we would have $c_{\APC_N(y),p}^\A = d^\A$, i.e.\ $c_{\APC_N(y),p}^\A$ and $c_{\APC_N(y),p}^\A + \varepsilon$ would be $\A$-covered by $\I_{\APC_N(y), N}$ -- a contradiction, as we have already argued then $p$ cannot be a point interval.
						Hence, there exists a constant symbol $e\in\bconsts(N)$ so that $e \neq d$ and $c_{\APC_N(y),p} = \alpha_{e+\varepsilon}$ and thus $p$ is either of the form $(e^\A, +\infty)$ or $(e^\A, r_u\rdbrace $ for some real value $r_u$. 
						In case of $d^\A \neq e^\A$, the instantiation point $\alpha_{d+\varepsilon}$ results in a part $p' \in \P_{\APC_N(y),N,\A}$ such that $\alpha_{d+\varepsilon}^\A \in p'$ and $p' \neq p$. But this contradicts our assumption $\alpha_\t^\A = \alpha_{e+\varepsilon}^\A \in p$.					
						Consequently, $d^\A = e^\A$. As $\widehat{N}_x$ contains $\Ax_{\{\alpha_{d+\varepsilon}\}, \{e\}}$ as a subset, $\A$ must be a hierarchic model of the axiom $d \equals e \to \alpha_{d+\varepsilon} \equals \alpha_{e+\varepsilon}$. This yields a contradiction, since above we concluded $\alpha_{d+\varepsilon}^\A = \alpha_\t^\A \neq c_{\APC_N(y),p}^\A = \alpha_{e+\varepsilon}^\A$.						 
			  	\end{description}

			\item Case $\A, \beta_\varphi \not\models s\approx s'$ for some free atom $s\approx s' \in \Gamma'$. Hence, $s$ and $s'$ are either variables or constant symbols of the free sort, which means they do not contain subterms of the base sort. Since $\B$ and $\A$ behave identical on free-sort constant symbols and $\beta(u) = \beta_\varphi(u)$ for any variable $u\in V_\S$, it must hold $\B, \beta \not\models s\approx s'$.

			\item Case $\A, \beta_\varphi \models s\approx s'$ for some $s\approx s' \in \Delta'$. Analogous to the above case, $\B, \beta \models s\approx s'$ holds.
			
			\item Case $\A, \beta_\varphi \not\models P(s_1, \ldots, s_m)$ for some free atom $P(s_1, \ldots, s_m)\in\Gamma'$. This translates to  $\bigl\<\A(\beta_\varphi)(s_1),$ $\ldots, \A(\beta_\varphi)(s_m)\bigr\> \not\in P^\A$.
				\begin{description}
					\item Every $s_i$ of the free sort is either a constant symbol or a variable. Thus, we have $\A(\beta_\varphi)(s_i) = \B(\beta)(s_i) = \varphi_{[\<P,i\>]}(\B(\beta)(s_i))$, since free-sort constant symbols are interpreted in the same way by $\A$ and $\B$, and because $\beta_\varphi(u) = \beta(u)$ for every free-sort variable $u$.					
					\item Every $s_i$ that is of the base sort must be a variable. Hence, $\A(\beta_\varphi)(s_i) = c_{[\<P,i\>],p}^\A = \varphi_{[\<P,i\>]}(\B(\beta)(s_i))$, where $p$ is the interval in $\P_{[\<P,i\>], N, \A}$ which contains $\beta(s_i)$ (and thus also $\beta_\varphi(s_i)$) and where we have $\APC_N(s_i) = [\<P,i\>]$.
				\end{description}
			Put together, this yields $\bigl\<\varphi_{[\<P,1\>]}(\B(\beta)(s_1)), \ldots, \varphi_{[\<P,m\>]}(\B(\beta)(s_m))\bigr\> \not\in P^\A$. But then, by construction of $\B$, we have $\bigl\<\B(\beta)(s_1), \ldots, \B(\beta)(s_m)\bigr\> \not\in P^\B$, which entails $\B, \beta \not\models P(s_1, \ldots, s_m)$.					
													
			\item Case $\A, \beta_\varphi \models P(s_1, \ldots, s_m)$ for some free atom $P(s_1, \ldots, s_m)\in\Delta'$. Analogous to the above case we conclude $\B, \beta \models P(s_1, \ldots, s_m)$.
		\end{description}
		Altogether, we have shown $\B\models N$.
\end{proof}

\subsection*{Details related to instantiation of Free-Sort Variables (Section \ref{section:BaseModelTheory})}

\begin{lemma}\label{lemma:FreeDomain}
	Let $N$ be a clause set in normal form and let $\A$ be a hierarchic model of $N$. We assume $\fconsts(N)$ to contain at least one constant symbol (otherwise we may add the tautology $\|\to c\approx c$ to $N$).
	By $\widehat{\S^\A}$ we denote the restricted domain $\{a \in \S^\A \mid \text{there is a $d\in\fconsts(N)$ such that $a = d^\A$}\}$.
	Let $\sim_\A$ be the binary relation on $\fconsts(N)$ satisfying $c \sim_\A d$ if and only if $c^\A = d^\A$.
	
	We can construct a hierarchic model $\B$ of $N$ which meets the following requirements:
	\begin{enumerate}[label=(\roman{*}), ref=(\roman{*})]
		\item\label{enum:FreeDomain:One} the domain $\S^\B$ is the set $\fconsts(N)/_{\sim_\A}$ of equivalence classes w.r.t.\ $\sim_\A$;
		\item\label{enum:FreeDomain:Two} $d^\B = [d]_{\sim_\A}$ for every free-sort constant symbol $d\in\fconsts(N)$;
		\item\label{enum:FreeDomain:Three} $c^\B = c^\A$ for every base-sort constant symbol $c\in\bconsts(N)\cup\aconsts(N)$; and
		\item\label{enum:FreeDomain:Four} for every atom $A$ in $N$ and every variable assignment $\beta : V_\S \cup V_\R \to \S^\A\cup \Real$ for which $\beta(u) \in \widehat{S^\A}$ for all $u\in V_\S$, we have $\A, \beta\models A$ if and only if $\B, \beta_{\sim_\A} \models A$, where for any variable $v \in V_\R\cup V_\S$ we set
			\[ \beta_{\sim_\A}(v) := 
					\begin{cases} 
						[d]_{\sim_\A}	&\text{if $v\in V_\S$ and $\beta(v) = d^\A$}\\
									&\text{for some $d\in\fconsts(N)$,}\\
						\beta(v)		&\text{if $v\in V_\R$.}
					\end{cases}
			\]
	\end{enumerate}
\end{lemma}
\begin{proof}
	We construct the hierarchic interpretation $\B$ as follows:
	\begin{itemize}
		\item $S^\B := \fconsts(N)/_{\sim_\A}$.
		\item For every constant symbol $d\in\fconsts(N)$, we set $d^\B := [d]_{\sim_\A}$.
		\item For every constant symbol $c\in\bconsts(N)\cup\aconsts(N)$, we set $c^\B := c^\A$.
	
		\item To help the formulation of the interpretation of predicate symbols under $\B$ we first define the function $\psi : \S^\B \cup \Real \to \S^\A \cup \Real$ by
			\[ \psi(a) :=
				\begin{cases}
					d^\A		&\text{if $a = [d]_{\sim_\A} \in \S^\B$ for some free-sort}\\ 
							&\text{constant symbol $d$,}\\
					a		&\text{if $a\in \Real$.}
				\end{cases}
			\]
			(Please note that this is well-defined, since $d \sim_\A d'$ entails $d^\A = d'^\A$.)
			
			For every predicate symbol $P/m$ that occurs in $N$ and for all arguments $a_1, \ldots, a_m$ of appropriate sorts, we define the interpretation of $P$ under $\B$ such that $\bigl\<a_1, \ldots, a_m\bigr\> \in P^{\B}$ if and only if $\bigl\< \psi(a_1), \ldots, \psi(a_m) \bigr\>\in P^\A$.
	\end{itemize}
	Obviously, requirements \ref{enum:FreeDomain:One}, \ref{enum:FreeDomain:Two} and \ref{enum:FreeDomain:Three} are satisfied.
	Due to the fact that we find a $\beta$ whose image is a subset of $\widehat{S^\A}\cup\Real$ (as described in \ref{enum:FreeDomain:Four}) for every variable assignment $\gamma : V_\R\cup V_\S \to \S^\B\cup\Real$ so that $\beta_{\sim_\A} = \gamma$, a proof of \ref{enum:FreeDomain:Four} entails $\B\models N$. Hence, it remains to show \ref{enum:FreeDomain:Four}.
	\begin{description}
		\item Suppose $A$ is an atomic constraint $s\triangleleft t$. The definitions of $\B$ and $\beta_{\sim_\A}$ immediately imply the equivalence of $\A,\beta\models A$ and $\B,\beta_{\sim_\A}\models A$.
		
		\item Suppose $A$ is of the form $s \approx t$, $s$ and $t$ being variables or constant symbols of the free sort, respectively. Because of $\A(\beta)(s), \A(\beta)(t)\in \widehat{S^\A}$, there exist constant symbols $d_s, d_t \in \fconsts(N)$ such that $d_s^\A = \A(\beta)(s)$ and $d_t^\A = \A(\beta)(t)$. Consequently, $\A(\beta)(s) = \A(\beta)(t)$ holds if and only if $d_s \sim_\A d_t$, i.e.\ if and only if $[d_s]_{\sim_\A} = [d_t]_{\sim_\A}$. Put differently, $\A, \beta\models s \approx t$ is equivalent to $\B,\beta_{\sim_\A}\models s \approx t$.

		\item Suppose $A$ is of the form $P(t_1, \ldots, t_m)$ for some predicate symbol $P/m$ occurring in $N$. As we have already argued in the previous case, we can find for every $t_i$ of the free sort a constant symbol $d_i$ fulfilling $d_i^\A = \A(\beta)(t_i)$. 
			\begin{description}
				\item If $t_i = c\in\fconsts(N)$, then $\psi(\B(\beta_{\sim_\A})(c)) = \psi([c]_{\sim_\A}) = c^\A = \A(\beta)(c)$.
				\item If $t_i = u\in V_\S$, then $\psi(\B(\beta_{\sim_\A})(u)) = \psi(\beta_{\sim_\A}(u)) = \psi([d_i]_{\sim_\A}) = d_i^\A = \beta(u) = \A(\beta)(u)$.
			\end{description}	
			
			On the other hand, for all $t_i$ of the base sort, we get $\psi(\B(\beta_{\sim_\A})(t_i)) = \B(\beta_{\sim_\A})(t_i) = \A(\beta)(t_i)$.
			
			Altogether, this leads to 
				\[ \bigl\< \psi\bigl(\B(\beta_{\sim_\A})(t_1) \bigr), \ldots, \psi\bigl(\B(\beta_{\sim_\A})(t_m)\bigr)\bigr\> = \bigl\<\A(\beta)(t_1), \ldots, \A(\beta)(t_m)\bigr\> \] 
			and, consequently, also to 
				\[ 
					\bigl\< \B(\beta_{\sim_\A})(t_1), \ldots, \B(\beta_{\sim_\A})(t_m) \bigr\> \in P^\B \;\text{ if and only if }\; \bigl\< \A(\beta)(t_1), \ldots, \A(\beta)(t_m) \bigr\> \in P^\A ~. 
					\qedhere
				\]
	\end{description}
\end{proof}

\begin{lemma}\label{lemma:EquisatisfiabilityFreeInstantiation}
	Let $N$ be a clause set in normal form that contains at least one constant symbol of the free sort.	
	Suppose there is a clause $C$ in $N$ which contains a free-sort variable $u$. Let the clause set $\widetilde{N}_u$ be constructed as follows:
	$\widetilde{N}_u := \bigl(N\setminus\{C\}\bigr)  \;\cup\; \bigl\{ C\subst{u}{c} \bigm| c\in\fconsts(N) \bigr\}$.
	The original clause set $N$ is satisfiable if and only if $\widetilde{N}_u$ is satisfiable. 
\end{lemma}
\begin{proof}
	While the ``only if''-part holds because of $\widetilde{N}_u$ containing only instances of clauses in $N$, the ``if''-part is slightly more complicated.  Let $\A$ be a hierarchic model of $\widetilde{N}_u$ so that $\S^\A = \fconsts(\widetilde{N}_u)/_\sim$ for some congruence relation $\sim$ on $\fconsts(N)$, and for every  free-sort constant symbol $c$ we have $c^\A = [c]_\sim$. But then $\A\models \bigl\{ C\subst{u}{c} \bigm| c\in\fconsts(N) \bigr\}$ entails $\A\models C$, since the set $\{[c]_\sim \mid c\in\fconsts(N)\}$ covers the whole domain $\S^\A$ due to $\fconsts(\widetilde{N}_u) = \fconsts(N)$. Hence, we obtain $\A\models N$.
\end{proof}

\subsection*{Details related to the Decision Procedure and its Complexity (Section~\ref{section:Complexity})}

\begin{lemma}\label{lemma:IterativeInstantiation}
	Let $N$ be a clause set in normal form that contains a clause $C$ in which a base-sort variable $x$ occurs. Further assume $\Ax_{\aconsts(N), \bconsts(N)} \subseteq N$ and for every $\alpha_{e+\varepsilon} \in \aconsts(N)$ we have $e\in\bconsts(N)$. Suppose $\widehat{N}_x$ is constructed from $N$ as described in Lemma \ref{lemma:EquisatisfiabilityBaseInstantiation} and variables have been renamed so that all clauses in $\widehat{N}_x$ are variable disjoint. 
	We observe the following facts:
	\begin{enumerate}[label=(\roman{*}), ref=(\roman{*})]
		\item\label{enum:IterativeInstantiation:One} ${\con_N} = {\con_{\widehat{N}_x}}$.
		\item\label{enum:IterativeInstantiation:Two} $\bconsts(N) = \bconsts(\widehat{N}_x)$ and $\fconsts(N) = \fconsts(\widehat{N}_x)$.
		\item\label{enum:IterativeInstantiation:Three} For every argument position class $[\<P,i\>]$ induced by $\con_{N}$ we have $\I_{[\<P,i\>], N} = \I_{[\<P,i\>], \widehat{N}_x}$.
		\item\label{enum:IterativeInstantiation:Four} $\widehat{N}_x$ is in normal form, $\Ax_{\aconsts(\widehat{N}_x), \bconsts(\widehat{N}_x)} \subseteq \widehat{N}_x$, and for every $\alpha_{e+\varepsilon} \in \aconsts(\widehat{N}_x)$ we have $e\in \bconsts(\widehat{N}_x)$.
	\end{enumerate}
\end{lemma}
\begin{proof}
	We prove the different parts separately.
	\medskip
	
	\noindent Ad \ref{enum:IterativeInstantiation:One}:
		The relations $\con_N$ and $\con_{\widehat{N}_x}$ exclusively depend on the free parts of the clauses in $N$ and $\widehat{N}_x$, respectively, irrespective of variable names. Since $\I_{\APC_N(x), N}$ is nonempty, all free parts of clauses in $N$ do occur in $\widehat{N}_x$ (modulo variable renaming), and vice versa.
	\medskip
	
	\noindent Ad \ref{enum:IterativeInstantiation:Two}:
		$\bconsts(N) \subseteq \bconsts(\widehat{N}_x)$ and $\fconsts(N) \subseteq \fconsts(\widehat{N}_x)$ hold due to $\I_{\APC_N(x), N}$ being nonempty. $\bconsts(\widehat{N}_x) \subseteq \bconsts(N)$ is a consequence of $\I_{\APC_N(x), N} \cap \Omega_\LA \subseteq \bconsts(N)$. $\fconsts(\widehat{N}_x) \subseteq \fconsts(N)$ is true, because free parts of clauses are merely copied and variables renamed.
	\medskip
	
	\noindent Ad \ref{enum:IterativeInstantiation:Three}:
		We start with two observations:
		\begin{enumerate}[label=(\arabic{*}),ref=(\arabic{*})]
			\item\label{enum:proofComplexity:Three:One} For every clause $D \neq C$ in $N$, we find a clause $\widehat{D}$ in $\widehat{N}_x$ such that $D$ and $\widehat{D}$ are the same modulo variable renaming. Conversely, every $\widehat{D}$ in $\widehat{N}_x$ is either identical to a clause $D$ in $N$ modulo renaming of variables or it is equal to $C\subst{x}{c}$ up to variable renaming for some constant symbol $c$.
		
			\item\label{enum:proofComplexity:Three:Two} Regarding $C$ we find the instance $C\subst{x}{\alpha_{-\infty}}$ in $\widehat{N}_x$, possibly with renamed variables.
				Consequently, for every base-sort variable $y\neq x$ for which $C$ contains an atomic constraint $y\triangleleft s$ and a free atom $Q(\ldots, y, \ldots)$ with $y$ in the $j$-th argument position, $C\subst{x}{\alpha_{-\infty}}$ contains the constraint $y\triangleleft s$ and the free atom $Q(\ldots, y, \ldots)$ with $y$ in the $j$-th argument position -- again, modulo variable renaming.
				The converse also holds.
		\end{enumerate}
		
		That said, we distinguish two cases for every argument position pair $\<P,i\>$.
		\begin{description}
			\item If $\<P,i\>$ does not belong to the equivalence class $\APC_N(x)$, i.e.\ $\APC_N(x) \neq [\<P,i\>]$, then 
				we get $\I_{[\<P,i\>],N} = \I_{[\<P,i\>],\widehat{N}_x}$, since \ref{enum:proofComplexity:Three:One} and \ref{enum:proofComplexity:Three:Two} together entail $\I_{Q,j,N} = \I_{Q,j,\widehat{N}_x}$ for all $\<Q,j\> \in [\<P,i\>]$.
			
			\item If $[\<P,i\>] = \APC_N(x)$, then there is an argument position pair $\<Q,j\> \in [\<P,i\>]$ so that $C$ contains a free atom $Q(\ldots, x, \ldots)$ in which $x$ is the $j$-th argument. Moreover, $\widehat{N}_x$ is a superset of $\bigl\{ C\subst{x}{c} \bigm| c\in\I_{\APC_N(x), N} \bigr\}$ modulo renaming of variables. This entails $\I_{\APC_N(x), N}\setminus\{\alpha_{-\infty}\} \subseteq \I_{Q,j,\widehat{N}_x}$, and therefore, $\I_{\APC_N(x), N} \subseteq \I_{[\<P,i\>], \widehat{N}_x}$.
			 
			 On the other hand, assume there is an instantiation point in $\I_{[\<P,i\>], \widehat{N}_x}$ that is not in $\I_{\APC_N(x), N}$, i.e.\ there is a clause $\widehat{D}$ in $\widehat{N}_x$ of which a variable-renamed variant has not already been in $N$, and in which an atomic constraint $y\triangleleft s$ and a free atom $Q(\ldots, y, \ldots)$ occur with $y$ in the $j$-th argument position. By \ref{enum:proofComplexity:Three:One}, $\widehat{D}$ must be a variable-renamed variant of $C\subst{x}{c}$ for some constant symbol $c$. But according to \ref{enum:proofComplexity:Three:Two}, we must have $(y\triangleleft s) = (x \equals c)$ modulo variable renaming. But this constraint can only lead to an instantiation point that is contained in $\I_{\APC_N(x), N}$ -- a contradiction. Hence, we also have $\I_{[\<P,i\>], \widehat{N}_x} \subseteq \I_{\APC_N(x), N}$. 
			 
			 Put together, we just derived $\I_{[\<P,i\>], N} = \I_{[\<P,i\>], \widehat{N}_x}$ for this case. 
		\end{description}
	  	
	\noindent Ad \ref{enum:IterativeInstantiation:Four}:
		The construction of $\widehat{N}_x$ from $N$ preserves the normal form property for each clause individually, since we start from a clause in normal form, copy the free part of the clause and no new variables are introduced to the constraint part which do not also occur in the free part. Afterwards, consistent renaming of variables ensures pairwise variable disjointness of the clauses in $\widehat{N}_x$.
	
		$\Ax_{\aconsts(\widehat{N}_x), \bconsts(\widehat{N}_x)} \subseteq \widehat{N}_x$ is a consequence of $\Ax_{\aconsts(N), \bconsts(N)} \subseteq N\setminus \{C\}$ (true by assumption and the fact that $C$'s free part is not empty but the free parts of axioms are) and $\Ax_{\I_{\APC_N(x), N}, \bconsts(N)} \subseteq \widehat{N}_x$ (true by construction). Moreover, $\Ax_{\aconsts(\widehat{N}_x), \bconsts(\widehat{N}_x)} = \Ax_{\aconsts(N), \bconsts(N)} \cup \Ax_{\I_{\APC_N(x), N}, \bconsts(N)}$ is ensured by \ref{enum:IterativeInstantiation:Two} and the fact that\\ $\aconsts(\widehat{N}_x) \subseteq \aconsts(N)\cup\I_{\APC_N(x),N}$.
		
		Suppose there is a constant symbol $\alpha_{e+\varepsilon}$ in $\widehat{N}_x$ so that $e$ does not occur as constant symbol in $\widehat{N}_x$. Hence, $\alpha_{e+\varepsilon}$ must have been an instantiation point in $\I_{\APC_N(x), N}$. Due to \ref{enum:IterativeInstantiation:Two}, $e$ cannot occur in $N$ either. Consequently, $\alpha_{e+\varepsilon}$ has been introduced to $\I_{\APC_N(x), N}$ by the occurrence of an atomic constraint $y \equals \alpha_{e+\varepsilon}$ in some clause in $N$. But we assumed $e' \in \bconsts(N)$ for every $\alpha_{e'+\varepsilon} \in \aconsts(N)$ -- a contradiction.
\end{proof}

\subsubsection*{Detailed proof of Theorem \ref{theorem:Complexity}}
\begin{proof}
	We devise a na\"ive algorithm that decides a given problem instance as follows.
	As input we assume a finite clause set $N$ such that $\aconsts(N) = \emptyset$ and $\fconsts(N)\neq\emptyset$.
	\begin{enumerate}[label=(\Roman{*}), ref=\Roman{*}]
		\item\label{enum:DecisionAlgorithm:One} Transform the input clause set $N$ into normal form by applying three steps to every clause $C = \Lambda \;\|\; \Gamma\to\Delta$ in $N$ that contains exactly $k > 0$ distinct base-sort variables $x^*_1, \ldots, x^*_k$ occurring in $\Lambda$ but not in $\Gamma\to\Delta$:
			\begin{enumerate}[label=(\ref{enum:DecisionAlgorithm:One}.\Roman{*}), ref=(\ref{enum:DecisionAlgorithm:One}.\Roman{*})]
				\item\label{enum:DecisionAlgorithm:OneOne} Let $\bar{y}$ be the vector of all variables $x^*_j$ among $x^*_1, \ldots, x^*_k$ for which there is an atomic constraint $x^*_j \equals c_j$ in $\Lambda$, and let $\bar{c}$ denote the corresponding vector of constant symbols $c_j$. We replace $C$ by $\Lambda\subst{\bar{y}}{\bar{c}}\;\|\;\Gamma\to\Delta$ in $N$.
				\item\label{enum:DecisionAlgorithm:OneTwo} If $\Lambda$ contains an atomic constraint of the form $x^*_j \not\equals c_j$ after Step \ref{enum:DecisionAlgorithm:OneOne}, then replace $C$ in $N$ by two clauses $\Lambda', x^*_j < c_j \;\|\; \Gamma\to\Delta$ and $\Lambda', x^*_j > c_j \;\|\; \Gamma\to\Delta$, where $\Lambda' := \Lambda \setminus \{x^*_j \not\equals c_j\}$. Iterate this procedure on the newly added clauses until all atomic constraints concerning one of the variables $x^*_1, \ldots, x^*_k$ have the form $x^*_j \triangleleft c_j$ with $\triangleleft \in \{<, \leq, \geq, > \}$.
				\item\label{enum:DecisionAlgorithm:OneThree} Apply Fourier-Motzkin quantifier elimination to every base-sort variable among\\ $x^*_1, \ldots,x^*_k$ that is left in $N$.
			\end{enumerate}	
			We call the resulting clause set $N'$. In addition, rename variables so that all clauses in $N'$ are variable disjoint.
		
		\item\label{enum:DecisionAlgorithm:Two} Compute the equivalence classes with respect to the equivalence relation $\con_{N'}$ induced by $N'$.
		
		\item\label{enum:DecisionAlgorithm:Three} Compute the set $\bconsts(N')$ of all base-sort constant symbols that occur in $N'$.
		
		 	For every equivalence class $[\<P,i\>]_{\con_{N'}}$ with $P:\xi_1\times\ldots\times\xi_m \in \Pi$ such that $\xi_i = \R$, collect all instantiation points that are connected to it, i.e.\ compute $\I_{[\<P,i\>], N'}$. Whenever a constant symbol $\alpha_\t$ freshly enters the collection of instantiation points, construct the axiom set $\Ax_{\{\alpha_\t\}, \bconsts(N')}$ according to Definition \ref{definition:BaseInstantiationAxioms} and add all these axioms to a maintained axiom set $\Ax$.
		
		\item\label{enum:DecisionAlgorithm:Four} Compute the set $\fconsts(N')$ of all free-sort constant symbols occurring in $N'$.
		
		\item\label{enum:DecisionAlgorithm:Five} Perform an all-at-once instantiation process on $N'$ that leads to 
			$N'' := \Bigl\{ C\subst{\overline{x}, \overline{u}}{\overline{c}, \overline{d}} \Bigm| C \in N' \text{ and}$ $ 
				 \{x_1, \ldots, x_k\} = \vars(C) \cap V_\R \text{ and }
				 \{u_1, \ldots, u_\ell\} = \vars(C) \cap V_\S \text{ and }
				 c_i\in\I_{\APC_{N'}(x_i), N'} \text{ and }$ $d_j\in\fconsts(N') \text{ for all $i$ and $j$} \Bigr\}$.

		\item\label{enum:DecisionAlgorithm:Six} 
			Nondeterministically construct a hierarchic interpretation $\A$ in three steps:
			\begin{enumerate}[label=(\ref{enum:DecisionAlgorithm:Six}.\Roman{*}), ref=(\ref{enum:DecisionAlgorithm:Six}.\Roman{*})]
				\item\label{enum:DecisionAlgorithm:SixOne} Nondeterministically choose a total ordering $\preceq_\A$ on the constant symbols in\\ $\bconsts(N'')\cup \aconsts(N'')$ so that for all $c,d \in \bconsts(N'')\cap\Real$ it holds $c\preceq_\A d$ if and only if $c \leq d$.
					Based on this ordering, construct (in a deterministic fashion) a mapping $\mu_\A : \bconsts(N'')\cup\aconsts(N'')\to\Real$ so that
					\begin{itemize}
						\item for all $c\in \bconsts(N'')\cap\Real$ we get $\mu_\A(c) = c$, and
						\item for all $d_1, d_2 \in \bconsts(N'')\cup\aconsts(N'')$ it holds $\mu_\A(d_1) \leq \mu_\A(d_2)$ if and only if $d_1 \preceq_\A d_2$.
					\end{itemize}
					Set $c^\A := \mu_\A(c)$ for every $c\in\bconsts(N'')$.
				
				\item\label{enum:DecisionAlgorithm:SixTwo} Nondeterministically choose a mapping $\nu_\A : \fconsts(N'') \to \fconsts(N'')$.
					
					While $\nu_\A$'s image shall induce the free domain $\S^\A := \{ \nu_\A(c) \mid c \in \fconsts(N'')\}$, the value $\nu_\A(c)$ shall be assigned to $c$ under $\A$ for every $c\in\fconsts(N'')$.

				\item\label{enum:DecisionAlgorithm:SixThree}	Let $\At(N'')$ denote the set of nonequational ground atoms induced by the essentially ground clause set $N''$, formally
						$\At(N'') := \bigl\{ A\subst{\bar{x}}{\bar{c}} \bigm| \text{ there is a non-equational free}$ 
						$\text{atom $A$ in the free part $\Gamma\to\Delta$ of a clause $(\Lambda, x_1\equals c_1, \ldots, x_k\equals c_k\;\|\;\Gamma\to\Delta)\in N''$}$\\
						$\text{with$\{x_1, \ldots, x_k\} = \vars(A)$} \bigr\}$.
					Given $\At(N'')$, we construct the set $\hAt(N'')$ by syntactically replacing every base-sort constant symbol $c$ in $\At(N'')$ by $\mu_\A(c)$ and every free-sort constant symbol $e$ by $\nu_\A(e)$.
					
					Nondeterministically choose a subset $\hAt_\A(N'')$ of $\hAt(N'')$ that represents the atoms in $\hAt(N'')$ which shall be true under $\A$, i.e.\ construct a Herbrand model over the atoms in $\hAt(N'')$.
			\end{enumerate}

		\item\label{enum:DecisionAlgorithm:Seven} Check whether $\A$ is a hierarchic model of $N''\cup \Ax$.
				\\[-2ex]\strut\hfill$\Diamond$
	\end{enumerate}
	
	Regarding the correctness of the algorithm, there are two crucial points which we need to address, namely (a) the equisatisfiability of $N$ and $N'$, and (b) the equisatisfiability of $N'$ and $N''\cup\Ax$. 

	\noindent Ad (a):
		It is straightforward to check that \ref{enum:DecisionAlgorithm:OneOne} and \ref{enum:DecisionAlgorithm:OneTwo} lead to equisatisfiable clause sets with respect to the standard semantics of linear arithmetic.
		
		The Fourier-Motzkin elimination step in \ref{enum:DecisionAlgorithm:OneThree} works for existentially quantified real-valued variables (cf.\ \cite{Dantzig1973}, \cite{Kroening2008}). Given a clause $C = \Lambda \;\|\; \Gamma\to\Delta$ that contains a base-sort variable $x$ occurring in $\Lambda$ but not in $\Gamma\to\Delta$, let $\bar{z}$ denote the vector of all variables in $\Lambda$ except for $x$. Recall that $C$ can be informally understood as $\forall \bar{z}\forall x \bigl(\bigl(\bigwedge \Lambda \wedge \bigwedge \Gamma\bigr) \to \bigvee \Delta\bigr)$. Since $x$ does neither occur in $\Delta$ nor in $\Gamma$ and since any formula $\phi\to\psi$ is equivalent to $\neg\phi \vee \psi$, we can equivalently write $\forall \bar{z} \bigl(\bigl((\exists x\bigwedge \Lambda) \wedge \bigwedge \Gamma\bigr) \to \bigvee \Delta\bigr)$.  Now Fourier-Motzkin variable elimination can be applied to the part $\exists x\bigwedge \Lambda$ to transform it into a formula $\bigwedge\Lambda'$ over linear arithmetic constraints (each of the form $c\triangleleft d$ or $y\triangleleft e$) so that $\bigwedge\Lambda'$ is equivalent to $\exists x\bigwedge\Lambda$ with respect to the standard semantics of linear arithmetic, $\Lambda'$ does not contain $x$ anymore and all variables and constant symbols in $\Lambda'$ have already occurred in $\Lambda$ (cf.\ \cite{Dantzig1973}). 
		
		Treating all base-sort variables that occur in $\Lambda$ but not in $\Gamma\to\Delta$ as described above, we finally obtain a clause $\Lambda'' \;\|\; \Gamma\to\Delta$ that is equivalent to $\Lambda \;\|\; \Gamma\to\Delta$ and is in normal form.
		\medskip		
		
	\noindent Ad (b):
		The construction of clause set $N''$ done in Step (\ref{enum:DecisionAlgorithm:Five}) resembles $|\vars(N')\cap V_\R|$ consecutive applications of Lemma \ref{lemma:EquisatisfiabilityBaseInstantiation} to the base-sort variables $x_1, \ldots, x_k$ followed by $|\vars(N')\cap V_\S|$ consecutive applications of Lemma \ref{lemma:EquisatisfiabilityFreeInstantiation} to the  free-sort variables $u_1,\ldots,u_\ell$ in order to instantiate all variables that occur in $N'$.
	
		For an arbitrary clause $C$ in $N'$ we have $\fconsts\bigl(C\subst{\bar{x}}{\bar{c}}\subst{u_1}{d_1}\ldots\subst{u_j}{d_j}\bigr) \subseteq \fconsts(N')$
		for every $j$, $0\leq j \leq \ell$. Hence, Lemma \ref{lemma:EquisatisfiabilityFreeInstantiation} entails equisatisfiability of $N'' \cup \Ax$ and the intermediate set $N''' \cup \Ax$ where
		$N''' := \bigl\{C\subst{\bar{x}}{\bar{c}} \bigm|
			 C \in N' \text{ and }
			 \{x_1, \ldots, x_k\} = \vars(C) \cap V_\R \text{ and } 
			 c_i \in\I_{\APC_{N'}(x_i), N'}$ $\text{ for all $i$} \bigr\}$.
			 
		It remains to show equisatisfiability of $N'$ and $N'''\cup \Ax$. 
		In order to do so, we first reduce the problem to showing equisatisfiability of $N'\cup \Ax$ and $N'''\cup \Ax$: If $N'$ is satisfiable, then so is $N'\cup\Ax$ by Proposition \ref{proposition:ImproperAxiomSatisfiability}, since $N'$ does not contain constant symbols from $\Omega_\LA^\alpha$. Conversely, $N'$ is obviously satisfiable whenever $N'\cup\Ax$ is.
	
		Let $x_1, \ldots, x_k$ be a list of pairwise distinct base-sort variables such that $\{x_1, \ldots, x_k\} := \vars(N')\cap V_\R = \vars(N'\cup\Ax)\cap V_\R$. We set $N'_0 := N'$ and for every $j$, $1\leq j\leq k$, we define $N'_j := \bigl\{C\subst{x_j}{c_j} \;\bigm|\; C\in N'_{j-1} \text{ and } c_j\in\I_{\APC_{N'}(x_j), N'}\bigr\}$.
		
		As the $x_1, \ldots, x_k$ are pairwise distinct and since the involved substitution operations only substitute variables with constant symbols, we can equivalently write	
		$N'_j =$\\ $\bigl\{C\subst{x_1}{c_1}\ldots\subst{x_j}{c_j} \bigm|  C \in N' \text{ and }  c_i \in\I_{\APC_{N'}(x_{i}), N'} \text{ and}$ $\text{$1\leq i\leq j$} \bigr\}$
		for every $j$, $1\leq j\leq k$. Variable disjointness of the clauses in $N'$ entails $N'_k = N'''$, since in this case the iterative substitution in the construction of $N'_k$ yields the same result as simultaneous substitution in the construction of $N'''$ does.
		
		Consider the sequence $N'_0, N'_1, \ldots, N'_k$ of clause sets, for which we know $N'_0 = N'$ and $N'_k = N'''$. If we rename the variables in each set in the sequence such that the clauses are variable disjoint, we can apply Lemma \ref{lemma:EquisatisfiabilityBaseInstantiation} $k$ times to conclude equisatisfiability of $N'_0\cup\Ax$, $N'_1\cup\Ax$ and so on up to $N'_k\cup\Ax$. To support this claim, we need to show that the prerequisites of Lemma \ref{lemma:EquisatisfiabilityBaseInstantiation} are fulfilled along the sequence. They certainly are for the starting point $N'_0\cup\Ax$, since we assume $N'$ to be in normal form and to not contain constant symbols $\alpha_\t$. For the rest of the sequence, we invoke Lemma \ref{lemma:IterativeInstantiation}, since it ensures ${\con_{N'_0\cup\Ax}} = \ldots = {\con_{N'_k\cup\Ax}}$ and $\I_{[\<P,i\>],N'_0} = \ldots = \I_{[\<P,i\>],N'_k}$ for all equivalence classes $[\<P,i\>]$ induced by $\con_{N'_0\cup\Ax}$.
		
		The following observations justify why it is legitimate for the algorithm described in steps (\ref{enum:DecisionAlgorithm:One}) to (\ref{enum:DecisionAlgorithm:Seven}) to mostly ignore the axioms in $\Ax$.
		\begin{enumerate}[label=(b.\arabic{*}), ref=(\ref{enum:proofComplexity:Two}.\arabic{*})]
			\item ${\con_{N'}} = {\con_{N'\cup\Ax}}$, because the relation $\con_{N'\cup\Ax}$ exclusively depends on the free parts of clauses in $N'\cup\Ax$, but the clauses in $\Ax$ do not have free parts.
			\item $\bconsts(N') = \bconsts(N'\cup\Ax)$, since the definition of $\Ax$ entails $\bconsts(\Ax)\subseteq \bconsts(N')$.
			\item $\fconsts(N') = \fconsts(N'\cup\Ax)$ holds because clauses in $\Ax$ do not contain constant symbols of the free sort.
			\item $\I_{[\<P,i\>],N'} = \I_{[\<P,i\>],N'\cup\Ax}$, since only atomic constraints that involve one base-sort variable contribute to the set of instantiation points, but $\Ax$ is a set of ground clauses.
		\end{enumerate}
			
		Altogether, the iterative application of Lemma \ref{lemma:EquisatisfiabilityBaseInstantiation} shows that $N'\cup \Ax$ has a hierarchic model if and only if $N'''\cup\Ax$ has one. But as we have already argued above, this entails the equisatisfiability of $N'$ and $N'''\cup\Ax$ and even $N''\cup\Ax$.
		\medskip
		
	Next, we investigate the running time of the presented algorithm. In order to do so, we take a look at every step of the algorithm individually.
	\smallskip

	\noindent Ad (\ref{enum:DecisionAlgorithm:One}):
		While the substitution operations in Step \ref{enum:DecisionAlgorithm:OneOne} take a total amount of time that is polynomial in $\len(N)$, the length of the clause set does not grow.
		The second step, however, may blow up the length of the clause set exponentially. Every clause $\Lambda\;\|\;\Gamma\to\Delta$ might be copied up to $2^{|\Lambda|}$ times, since $\Lambda$ can contain at most $|\Lambda|$ constraints of the form $x_j\not\equals c_j$. Hence, Step \ref{enum:DecisionAlgorithm:OneTwo} increases the length of the clause set to not more than $\len(N)\cdot 2^{\len(N)}$ and in the worst case takes time polynomial in that new length.
		Given a multiset $\Lambda$ of atomic constraints, in which variables $x_1, \ldots, x_k \in \vars(\Lambda)$ are supposed to be eliminated one after another by the Fourier-Motzkin procedure, we can partition $\Lambda$ into $k+1$ parts $\Lambda_0, \Lambda_1, \ldots, \Lambda_k$ so that for all $i > 0$ the part $\Lambda_i$ contains all atomic constraints that involve $x_i$ and only those, and $\Lambda_0 := \Lambda \setminus (\Lambda_1 \cup \ldots \cup \Lambda_k)$. It turns out that eliminating a variable $x_i$ from $\Lambda$ results in a multiset $(\Lambda \setminus \Lambda_i) \cup \widehat{\Lambda}_i$ where $|\widehat{\Lambda}_i|$ is at most $|\Lambda_i|^2$. Hence, after eliminating all $x_i$, we end up with a multiset of atomic constraints of size $|\Lambda_0| + |\Lambda_1|^2 + \ldots + |\Lambda_k|^2 \leq |\Lambda|^2$. Consequently, Step \ref{enum:DecisionAlgorithm:OneTwo} may increase the length of the clause set at most quadratically. The time taken for this step is polynomial in that new length.
		Overall, we end up with a length of at most $\len(N)^2 \cdot 2^{2\cdot\len(N)} \leq \len(N)^{3\cdot\len(N)}$ for $N'$ (recall that we assumed $\len(N) \geq 2$).

	\noindent Ad (\ref{enum:DecisionAlgorithm:Two}):
		This step can be performed in time that is polynomial in the length of $N'$ using an efficient union-find data structure.
		
	\noindent Ad (\ref{enum:DecisionAlgorithm:Three}):
		The computation of $\bconsts(N')$ and the collection of all relevant instantiation points takes time polynomial in the length of $N'$. 
		The set of instantiation points for any base-sort variable is a subset of $\bconsts(N') \cup \{\alpha_{d+\varepsilon} \;\mid\; d\in\bconsts(N')\} \cup \{\alpha_{-\infty}\}$. It is worthwhile to note that Steps \ref{enum:DecisionAlgorithm:OneOne} to \ref{enum:DecisionAlgorithm:OneThree} do not lead to a change in the number of instantiation points, since they only modify atomic constraints that do not contribute to instantiation points. The reason is that the variables $x_1, \ldots, x_k$ addressed in Step (\ref{enum:DecisionAlgorithm:One}) do not occur in the free parts of the modified clause. 
		For every argument position pair $\<P,i\>$, there are at most $\len(N)$ instantiation points, as every atomic constraint $y\triangleleft c$ in $N$ can induce at most one instantiation point (either $c$ of $\alpha_{c+\varepsilon}$). To account for $\alpha_{-\infty}$: if there is a base-sort variable $y$ to be instantiated in $N'$ at all, then $N'$ must also contain a free atom $Q(\ldots,y,\ldots)$ which did already occur in $N$, and which thus also contributes to the length of $N$. In addition, we have $\bconsts(N') \subseteq \bconsts(N)$, leading to $|\bconsts(N')| \leq \len(N)$.
		The construction of the required axiom set $\Ax$ can be done in polynomial time in $\len(N') + \len(\Ax)$, where we can bound the length of the axiom set from above by $2\cdot2\cdot\len(N) + 2\cdot 4\cdot\len(N)^2 \leq 10 \cdot \len(N)^2$.
				
	\noindent Ad (\ref{enum:DecisionAlgorithm:Four}):
		The extraction of $\fconsts(N')$ does not take longer than polynomial time in the length of $N'$.

	\noindent Ad (\ref{enum:DecisionAlgorithm:Five}):
		At first, we consider each clause $C = \Lambda\;\|\;\Gamma\to\Delta$ in $N'$ separately. Since $|\fconsts(N')|$ is upper bounded by $\len(N)$ (every free-sort constant symbol in $N'$ did already occur in $N$), instantiation of the free-sort variables yields a factor of at most\\ $\len(N)^{|\vars(\Gamma\to\Delta) \cap V_\S|}$. We have already argued -- when looking at Step (\ref{enum:DecisionAlgorithm:Three}) -- that the number of instantiations points for each base-sort variable is bounded from above by $\len(N)$. Hence, instantiation of all base-sort variables in the clause adds a factor of at most\\ $\len(N)^{|\vars(\Gamma\to\Delta) \cap V_\R|}$.
		There are at most $\len(\Gamma\to\Delta)$ different variables in $C$ that need to be instantiated, and as $\Gamma\to\Delta$ did already occur in $N$ (modulo variable renaming), we have $\len(\Gamma\to\Delta) \leq \len(N)$.  				   
		When instantiating a clause, the constraint part may increase in length, namely by at most double the number of instantiated variables. In the worst case, we thus get triple the length of the original, e.g.\ in case of instantiating the clause $\|\to P(x)$ with the default instantiation point $\alpha_{-\infty}$ we obtain $x\equals \alpha_{-\infty}\;\|\;\to P(x)$.
		In total, instantiating a single clause $C = \Lambda\;\|\;\Gamma\to\Delta$ taken from $N'$ leads to a clause set of length at most  $3\cdot\len(C) \cdot \len(N)^{\len(\Gamma\to\Delta)}$. Consequently, we can upper bound the length of the fully instantiated clause set $N''$ by $3\cdot\len(N') \cdot \len(N)^{\len(N)} \;\leq\;
			3\cdot  \len(N)^{4\cdot\len(N)}$.
		Instantiating the set of clauses needs only time that is bounded by some polynomial in $\len(N'')$.

	\noindent Ad (\ref{enum:DecisionAlgorithm:Six}):
		The construction of $\A$ can be done nondeterministically in time that is bounded from above by some polynomial in the length of $N''$.

	\noindent Ad (\ref{enum:DecisionAlgorithm:Seven}):
		The check whether $\A$ satisfies $N''\cup\Ax$ can be performed in a deterministic fashion in time polynomial in the length of $N''\cup\Ax$.
		\smallskip

		Taking all the above results into account, we can upper bound the running time of the algorithm by some polynomial in $\len(N)^{4\cdot\len(N)}$. Hence, there is some constant $\mathfrak{c} \geq 4$ such that the nondeterministic running time lies in $\O\bigl(\len(N)^{\mathfrak{c}\cdot \len(N)}\bigr)$. Consequently, the problem of deciding whether a finite clause set $N$ is satisfiable, lies in NEXPTIME.
\end{proof}


\end{document}